\documentclass[11pt, letterpaper]{article}
\usepackage{amssymb}
\usepackage{fullpage}
\usepackage[noadjust]{cite}
    
\title{Dual Lower Bounds for 
Approximate Degree and Markov-Bernstein Inequalities}

\author{Mark Bun\thanks{Harvard University, School of Engineering and Applied Sciences. Supported by an NDSEG Fellowship and NSF grant CNS-1237235.}
   \and 
Justin Thaler\thanks{Simons Institute for the Theory of Computing at UC Berkeley. This work was performed while the author was a graduate student at Harvard University, School of Engineering and Applied Sciences, and supported by an NSF Graduate Research Fellowship and NSF grants CNS-1011840 and CCF-0915922.}   }

\date{}

\usepackage{amsfonts,amsmath,color,float,graphicx,verbatim,url}

\usepackage{enumerate}

\newenvironment{proof}[1][Proof: ]{\noindent \textbf{#1}}{\qed\medskip}
\newenvironment{proofof}[1]{\noindent \textbf{Proof of #1:}}{\qed\medskip}

\newcommand{\R}{\mathbb{R}}

\newcommand{\sgn}{\widetilde{\mathrm{sgn}}}

\usepackage{color}

\newcommand{\ignore}[1]{}
\newcommand{\provisionallyremove}[1]{}

\newcommand{\eps}{\vareps}
\newcommand{\vareps}{\varepsilon}

\newcommand{\subsecref}[1]{Subsection~\ref{#1}}

\renewcommand{\eqref}[1]{Eq.~(\ref{#1})}
\newcommand{\lemref}[1]{Lemma~\ref{#1}}
\newcommand{\corref}[1]{Corollary~\ref{#1}}
\newcommand{\thmref}[1]{Theorem \ref{#1}}
\newcommand{\propref}[1]{Proposition~\ref{#1}}
\newcommand{\appref}[1]{Appendix~\ref{#1}}

\newtheorem{theorem}{Theorem}
\newtheorem{lemma}[theorem]{Lemma}
\newtheorem{proposition}[theorem]{Proposition}
\newtheorem{corollary}[theorem]{Corollary}

\newtheorem{fact}[theorem]{Fact}
\newtheorem{remark}[theorem]{Remark}

\newcommand{\qed}{\hfill\rule{7pt}{7pt} \medskip}

\begin{document}
 
\maketitle

\begin{abstract}
The $\eps$-\emph{approximate degree} of a Boolean function $f: \{-1, 1\}^n \to \{-1, 1\}$ is the minimum degree of a real polynomial that approximates $f$ to within error $\eps$ in the $\ell_\infty$ norm. We prove several lower bounds on this important complexity measure by explicitly constructing solutions to the dual of an appropriate linear program.
 Our first result resolves the $\eps$-approximate degree of the two-level AND-OR tree for any constant $\eps > 0$. 
 We show that
 this quantity is $\Theta(\sqrt{n})$, closing a line of incrementally larger lower bounds \cite{ambainis, hoyer, NS94, sherstovFOCS, shi}.
 The same lower bound was recently obtained independently by Sherstov using related techniques \cite{sherstovnew}.
 Our second result gives an explicit \emph{dual polynomial} that witnesses a tight lower bound for the approximate degree of any symmetric Boolean function, addressing a question of \v{S}palek \cite{spalek}. Our final contribution is to reprove several Markov-type inequalities from approximation theory by constructing explicit dual solutions to natural linear programs. These inequalities underly the proofs of many of the best-known approximate degree lower bounds, and have important uses throughout theoretical computer science.
\end{abstract}\medskip 

\section{Introduction} 
Approximate degree is an important measure of the complexity of a Boolean function. It captures whether
a function can be approximated by a low-degree polynomial with real coefficients in the $\ell_{\infty}$ norm, and it has diverse applications in theoretical computer science.
For instance, lower bounds on approximate degree underly fundamental circuit complexity lower bounds \cite{mp, beigel93, sherstovseparate} and oracle separations between complexity classes \cite{beigelomb}. In quantum computing, many tight lower bounds on quantum query complexity have been proved via lower bounds on approximate degree \cite{aaronsonshi, klauck, beals}. Approximate degree lower bounds have also found important uses in communication complexity
\cite{patternmatrix, patternmatrixfollowup, ada, spalek, bvw, shizhu, sherstovmultidisj}, enabling the resolution of long-standing open problems regarding both randomized and quantum formulations of bounded-error, small-bias, and multiparty communication. Meanwhile, upper bounds on approximate degree have had several important algorithmic uses. For instance, in computational learning theory, approximate degree upper bounds
underly the best known algorithms for PAC learning DNF and read-once formulas, and agnostically learning disjunctions \cite{klivansservedio, learningreadonce, agnostic}.  


In this paper, we seek to advance our understanding of this fundamental complexity measure. 
We focus on proving approximate degree lower bounds by specifying explicit \emph{dual polynomials}, which are dual solutions to a certain 
linear program capturing the approximate degree of any function. These polynomials act as certificates of the high approximate degree of a function, and 
their construction is of interest because these dual objects have been used recently
to resolve several long-standing open problems in communication complexity (e.g. \cite{patternmatrix, patternmatrixfollowup, ada, spalek, bvw, shizhu}). See
the survey of Sherstov \cite{sherstovsurvey} for an excellent overview of this body of literature. 

\medskip
\noindent \textbf{Our Contributions.}
Our first result resolves the approximate degree of the function $f(x) = \wedge_{i=1}^N \vee_{j=1}^N x_{ij}$, showing this quantity is $\Theta(N)$.
Known as the two-level AND-OR tree, $f$ is perhaps the simplest function whose approximate degree was not previously characterized.
A series of works spanning nearly two decades proved incrementally larger lower bounds on the approximate degree of this function, and this question
was recently re-posed by Aaronson in a tutorial at FOCS 2008 \cite{scottslides}. 
Our proof not only yields a tight lower bound, but it specifies an explicit dual polynomial for the high approximate degree of $f$,
answering a question of \v{S}palek  \cite{spalek} in the affirmative.

Our second result gives an explicit dual polynomial witnessing the high approximate degree of any \emph{symmetric} Boolean function,
recovering a well-known result of Paturi \cite{paturi}. Our solution builds on work of \v{S}palek \cite{spalek},
who gave an explicit dual polynomial for the OR function, and addresses an open question from that work.

Our final contribution is to reprove several classical Markov-type inequalities from approximation theory. 
These inequalities bound the derivative of a polynomial in terms of its degree. Combined with the well-known symmetrization technique (see e.g. \cite{mp, scottslides}), Markov-type inequalties 
have traditionally been the primary tool used to prove approximate degree lower bounds on Boolean functions (e.g. \cite{aaronsonshi, ambainis, NS94, shi}). 
Our proofs of these inequalities specify explicit dual solutions to a natural linear program (that differs from the one used to prove our first two results). 
While these inequalities have been known for over a century \cite{bernstein, markov1, markov2}, to the best of our knowledge our proof technique is novel, and we believe it sheds new light on these results.

\section{Preliminaries}
We work with Boolean functions $f: \{-1, 1\}^n \to \{-1, 1\}$
under the standard convention that 1 corresponds to logical false, and $-1$ corresponds to logical true. We let $\|f\|_\infty = \max_{x \in \{-1, 1\}^n} |f(x)|$ denote the $\ell_\infty$ norm of $f$. 
The $\eps$-approximate degree of a function $f: \{-1, 1\}^n \rightarrow \{-1, 1\}$,
denoted $\deg_\eps(f)$, is the minimum (total) degree of any real polynomial $p$ such that $\|p - f\|_\infty \le \eps$, i.e., $|p(x) - f(x)| \leq \eps$ for all $x \in \{-1, 1\}^n$.
We use $\widetilde{\deg}(f)$ to denote $\deg_{1/3}(f)$, and use this to refer to the \emph{approximate degree} of a function without qualification. The choice of $1/3$ is arbitrary, as $\widetilde{\deg}(f)$ is related to $\deg_\eps(f)$ by a constant factor for any constant $\eps \in (0, 1)$. We let OR$_n$ and AND$_n$ denote the OR function and AND function on $n$ variables respectively,
and we let $\mathbf{1}_{n} \in \{-1, 1\}^{n}$ denotes the $n$-dimensional all-ones vector.  Define $\sgn(x) = -1$ if $x < 0$ and 1 otherwise.


In addition to approximate degree, \emph{block sensitivity} is also an important measure of the complexity of a Boolean function. We introduce this measure because functions with low block sensitivity are an ``easy case'' in the analysis of Theorem \ref{thm:andor} below. The block sensitivity $\text{bs}_x(f)$ of a Boolean function $f: \{-1, 1\}^n \rightarrow  \{-1, 1\}$ at the point $x$ is the maximum number of pairwise disjoint subsets $S_1, S_2, S_3, \dots \subseteq \{1, 2, \dots , n\}$ such that $f(x) \neq f(x^{S_1}) = f(x^{S_2}) = f(x^{S_3}) = \dots  $ Here, $x^S$ denotes the vector obtained from $x$ by negating each entry whose index is in $S$. The block sensitivity $\operatorname{bs}(f)$ of $f$ is the maximum of $\operatorname{bs}_x(f)$ over all $x \in \{-1, 1\}^n.$

\subsection{A Dual Characterization of Approximate Degree}
\label{sec:dualcharacterization}
For a subset $S \subseteq \{1, \dots, n\}$ and $x \in \{-1, 1\}^n$, let $\chi_S(x) = \prod_{i \in S} x_i$. Given a Boolean function $f$, let $p(x) = \sum_{|S| \leq d} c_S \chi_S(x)$ be a polynomial of degree $d$ that minimizes $\|p - f\|_\infty$, where the coefficients $c_S$ are real numbers. Then $p$ is an optimum of the following linear program.

\[ \boxed{\begin{array}{lll} 
    \text{min}  &     \eps    \\
    \mbox{such that} &\Big|f(x) - \sum_{|S| \leq d} c_S \chi_S(x)\Big| \leq \eps & \text{ for each } x \in \{-1, 1\}^n\\
    &c_S \in \mathbb{R} & \text{ for each } |S| \leq d\\
    &\eps \geq 0
    \end{array}}
\]

The dual LP is as follows.

\[ \boxed{\begin{array}{lll} 
    \text{max} &    \sum_{x \in \{-1, 1\}^n} \phi(x) f(x)   \\
    \mbox{such that} &\sum_{x \in \{-1, 1\}^n} |\phi(x)| = 1\\
    &\sum_{x \in \{-1, 1\}^n} \phi(x) \chi_S(x)=0  & \text{ for each } |S| \leq d\\
    &\phi(x) \in \mathbb{R} & \text{ for each } x \in \{-1, 1\}^n
    \end{array}}
\]

Strong LP-duality yields the following well-known dual characterization of approximate degree (cf. \cite{patternmatrix}).

\begin{theorem} \label{thm:prelim} Let $f: \{-1, 1\}^n \to \{-1, 1\}$ be a Boolean function. Then $\deg_\eps(f) > d$ if and only if there is a polynomial $\phi: \{-1, 1\}^n \rightarrow \mathbb{R}$ such that 
\begin{equation} \label{eq:prelim0} \sum_{x \in \{-1, 1\}^n} f(x) \phi(x) > \eps, \end{equation}
\begin{equation} \label{eq:prelim1} \sum_{x \in \{-1, 1\}^n} |\phi(x)| = 1,\end{equation}  and
\begin{equation} \label{eq:prelim2} \sum_{x \in \{-1, 1\}^n} \phi(x) \chi_S(x)=0   \text{ for each } |S| \leq d.\end{equation}
\end{theorem}

If $\phi$ satisfies \eqref{eq:prelim2}, we say $\phi$ has \emph{pure high degree} $d$.
We refer to any feasible solution $\phi$ to the dual LP as a \emph{dual polynomial} for $f$.

\section{A Dual Polynomial for the AND-OR Tree}
Define $\operatorname{AND-OR}^{M}_{N}: \{-1, 1\}^{MN} \rightarrow \{-1, 1\}$ by $f(x) = \wedge_{i=1}^{M} \vee_{j=1}^{N} x_{ij}$.
$\operatorname{AND-OR}^{N}_{N}$ is known as the two-level AND-OR tree, and its approximate degree has resisted characterization for close to two decades.
Nisan and Szegedy proved an $\Omega(N^{1/2})$ lower bound on $\widetilde{\deg}(\text{AND-OR}^{N}_{N})$ in \cite{NS94}.
This was subsequently improved to $\Omega(\sqrt{N \log N})$ by Shi \cite{shi}, and improved further to $\Omega(N^{2/3})$ by Ambainis \cite{ambainis}.
Most recently, Sherstov proved an $\Omega(N^{3/4})$ lower bound in \cite{sherstovFOCS}, which was the best lower bound prior to our work. 
The best upper bound is $O(N)$ due to H\o yer, Mosca, and de Wolf \cite{hoyer}, which matches our new lower bound. 

By refining Sherstov's analysis in \cite{sherstovFOCS}, we will show that $\widetilde{\deg}(\text{AND-OR}^{M}_{N}) = \Omega(\sqrt{MN})$, which matches an upper bound 
implied by a result of Sherstov \cite{sherstovrobust}. In particular, this implies that the approximate degree of the two-level AND-OR tree is $\Theta(N)$.

\begin{theorem} \label{thm:andor} $\widetilde{\deg}(\operatorname{AND-OR}^{M}_{N}) = \Theta(\sqrt{MN})$. \end{theorem}

\medskip
\noindent \textbf{Independent work by Sherstov.} Independently of our work, Sherstov \cite{sherstovnew} has discovered the same $\Omega(\sqrt{MN})$ lower bound 
on $\widetilde{\deg}(\text{AND-OR}^{M}_{N})$. Both his proof and ours exploit the fact that the OR function has a dual polynomial with one-sided error. 
Our proof proceeds by constructing an explicit dual polynomial for $\text{AND-OR}^{M}_{N}$, by combining a dual polynomial
for OR$_N$ with a dual polynomial for AND$_M$. In contrast, Sherstov mixes the primal and dual views: his proof combines a dual polynomial for OR$_{N}$
with an approximating polynomial $p$ for $\text{AND-OR}^{M}_{N}$ to construct an approximating polynomial $q$ for $\text{AND}_{M}$. The proof in \cite{sherstovnew} shows that
$q$ has much lower degree than $p$, so the desired lower bound on the degree of $p$ follows from known lower bounds on the degree of $q$.

The proof of \cite{sherstovnew} is short (barely more than a page), while our proof has the benefit of yielding an explicit dual polynomial witnessing the lower bound.

\subsection{Proof Outline} \label{sec:and-or-outline}
Our proof is a refinement of a result of Sherstov \cite{sherstovFOCS}, which roughly showed that approximate degree increases multiplicatively under function composition. 
Specifically, Sherstov showed the following.

\begin{proposition}[{\cite[Theorem 3.3]{sherstovFOCS}}] \label{prop:sherstov} Let $F: \{-1, 1\}^{M} \rightarrow \{-1, 1\}$ and $f: \{-1, 1\}^{N} \to \{-1, 1\}$ be 
given functions. Then for all $\eps, \delta > 0$,
 $$\deg_{\eps - 4 \delta \operatorname{bs}(F)}(F(f, \dots, f)) \geq \deg_\eps(F) \deg_{1-\delta}(f).$$ \end{proposition}

Sherstov's proof of Proposition \ref{prop:sherstov} proceeds by taking a dual witness $\Psi$ to the high $\eps$-approximate degree of $F$, and combining it with a dual witness $\psi$ to the high $(1-\delta)$-approximate degree of $f$
to obtain a dual witness $\zeta$ for the high $(\eps - 4 \delta \text{bs}(F))$-approximate degree of $F(f, \dots, f)$. 
His proof proceeds in two steps: he first shows that $\zeta$ has pure-high degree at least  $\deg_\eps(F) \deg_{1-\delta}(f)$,
and then he lower bounds the correlation of $\zeta$ with $F(f, \dots, f)$. The latter step of this analysis yields a lower bound
on the correlation of $\zeta$ with $F(f, \dots, f)$ that deteriorates rapidly as the block sensitivity $\text{bs}(F)$ grows.

Proposition \ref{prop:sherstov} itself does not yield a tight lower bound for $\widetilde{\text{deg}}(\text{AND-OR}^{M}_{N})$, because
the function $\text{AND}_{M}$ has maximum block sensitivity $\text{bs}(\text{AND}_{M})=M$. We address this by refining the second step of Sherstov's analysis
in the case where $F=\text{AND}_{M}$ and $f=\text{OR}_{N}$. We leverage two facts. First, although the block sensitivity of $\text{AND}_{M}$ is high,
it is only high at one input, namely the all-true input. At all other inputs, $\text{AND}_{M}$ has low block sensitivity and the analysis of Proposition \ref{prop:sherstov} is tight. Second, we use the fact that any dual witness to the high approximate degree of $\text{OR}_N$ has one-sided error. Namely, if $\psi(x)<0$ for such a dual witness $\psi$, 
then we know that $\psi(x)$ agrees in sign with $\text{OR}_{N}(x)$. This property allows us to handle the all-true input to $\text{AND}_{M}$ separately: we use it to show that despite
the high block-sensitivity of $\text{AND}_M$ at the all-true input $y$, this input nonetheless contributes positively to the correlation between $\zeta$ and $F(f, \dots, f)$. The details of our construction follow.

\subsection{Proof of \thmref{thm:andor}}
\label{sec:andorproof}
\vspace{-4mm}
\begin{proof}[]

Nisan and Szegedy \cite{NS94} proved the now well-known result that for any constant $0 < \eps<1$, 
$\deg_\eps(\text{AND}_n) = \deg_\eps(\text{OR}_n) = \Theta(\sqrt{n})$. Let $\Psi: \{-1, 1\}^{M} \rightarrow \mathbb{R}$ be a dual witness for the $(1/3)$-approximate degree of 
$\text{AND}_{M}$ whose existence is guaranteed by Theorem 
\ref{thm:prelim}. There is some $\eps > 1/3$ and $d = \Theta(\sqrt{M})$ such that $\Psi$ satisfies:
\begin{equation} \label{eqPsi1} \sum_{x \in \{-1, 1\}^{M}} \Psi(x) \text{AND}_{M}(x) = \eps, \end{equation}
\begin{equation} \label{eqPsi2} \sum_{x \in \{-1, 1\}^{M}} |\Psi(x)| = 1,\end{equation}
\begin{equation} \label{eqPsi3} \sum_{x \in \{-1, 1\}^{M}} \Psi(x) \chi_S(x)=0   \text{ for each } |S| \leq d.\end{equation}

Likewise, let $\psi$ be the dual witness for the $(1-(\eps - 1/3)/4)$-approximate degree of $\text{OR}_{N}$. By \thmref{thm:prelim}, there is some $\delta < (\eps - 1/3)/4$ and some $d'=\Theta(\sqrt{N})$ such that $\psi$ satisfies:
\begin{equation} \label{eqpsi1} \sum_{x \in \{-1, 1\}^{N}} \psi(x) \text{OR}_{N}(x) = 1-\delta, \end{equation}
\begin{equation} \label{eqpsi2} \sum_{x \in \{-1, 1\}^{N}} |\psi(x)| = 1,\end{equation}
\begin{equation} \label{eqpsi3} \sum_{x \in \{-1, 1\}^{N}} \psi(x) \chi_S(x) =0   \text{ for each } |S| \leq d'.\end{equation}

We will also make use of the following easy lemma, which tells us the precise values
of $\psi(\mathbf{1}_{N})$ and $\Psi(-\mathbf{1}_{M})$. This is essentially a restatement of 
a result due to Gavinsky and Sherstov \cite{gavinskysherstov}.

\begin{lemma}\label{facts}
\begin{equation} \label{eq:facts1} 1-\delta = \sum_{x \in \{-1, 1\}^{N}} \psi(x) \operatorname{OR}_{N}(x) = 2 \psi(\mathbf{1}_{N}).\end{equation} In particular, $\psi(\mathbf{1}_{N}) > 0$. Similarly,
\begin{equation} \label{eq:facts2} \eps = \sum_{x \in \{-1, 1\}^{M}} \Psi(x) \operatorname{AND}_{M}(x) = -2 \Psi(-\mathbf{1}_{M}).\end{equation} In particular, $\Psi(-\mathbf{1}_{M}) < 0$. \end{lemma}

\begin{proof}[Proof of Lemma \ref{facts}:]
The first part follows because
\[\sum_{x \in \{-1, 1\}^N} \psi(x) \text{OR}_{N}(x)  = 2 \psi(\mathbf{1}) - \sum_{x \in \{-1, 1\}^N} \psi(x).\]
The second term on the right-hand side is zero because $\psi$ is orthogonal 
to all polynomials of degree at most $d$, and in particular $\psi$ is orthogonal to the constant function. The proof for the second part is
similar.\end{proof}

As in Sherstov's proof of Proposition \ref{prop:sherstov}, we define
$\zeta: \left(\{-1, 1\}^{N}\right)^{M} \rightarrow \mathbb{R}$ by
\begin{equation} \label{eq:zeta} \zeta(x_1, \dots, x_{M}) := 2^{M} \Psi(\dots, \sgn(\psi(x_i)), \dots) \prod_{i=1}^{M} |\psi(x_i)|,\end{equation}
where $x_i = (x_{i, 1}, \dots, x_{i, N})$.

By \thmref{thm:prelim}, in order to show that $\zeta$ is a dual witness for the fact that the $(1/3)$-approximate degree of $\text{AND-OR}^M_N$
is $\Omega(\sqrt{MN})$, it suffices to show that 

\begin{equation} \label{eq:show1} \sum_{(x_1, \dots, x_{M}) \in \left(\{-1, 1\}^{N}\right)^M} \zeta(x_1, \dots, x_{M}) \text{AND-OR}^{M}_{N}(x_1, \dots, x_{M}) \geq 1/3.\end{equation}
\begin{equation} \label{eq:show2} \sum_{(x_1, \dots, x_{M}) \in \left(\{-1, 1\}^{N}\right)^M}  |\zeta(x_1, \dots, x_{M})| = 1. \end{equation}
\begin{equation} \label{eq:show3}  \sum_{(x_1, \dots, x_{M}) \in \left(\{-1, 1\}^{N}\right)^M}  \zeta(x_1, \dots, x_{M})\chi_S(x_1, \dots, x_M) =0   \text{ for each } |S| \leq d\cdot d'.\end{equation}

\eqref{eq:show3} is proved exactly as in \cite{sherstovFOCS}; we provide Sherstov's argument in Appendix \ref{app:andor2} for completeness. We now argue that Expression (\ref{eq:show1}) and \eqref{eq:show2} hold as well.

\begin{proof}[Proof of \eqref{eq:show2}.] Let $\mu$ be the distribution on $\left(\{-1, 1\}^{N}\right)^M$ given by $\mu(x_1, \dots, x_M) = \prod_{i=1}^{M} |\psi(x_i)|$. Since $\psi$ is orthogonal to the constant polynomial, it has expected value 0, and hence the string $(\dots, \sgn(\psi(x_i)), \dots)$ is distributed uniformly in $\{-1, 1\}^{M}$
when one samples $(x_1, \dots, x_{M})$ according to $\mu$. 
Thus,
$$ \sum_{(x_1, \dots, x_{M}) \in \left(\{-1, 1\}^{N}\right)^M} |\zeta(x_1, \dots, x_{M})| = \sum_{z \in \{-1, 1\}^M} |\Psi(z)| = 1$$
by 
\eqref{eqPsi2}, proving \eqref{eq:show2}. 
\end{proof}

\begin{proof}[Proof of Expression (\ref{eq:show1}).] Using the same distribution $\mu$ as in the proof of \eqref{eq:show2}, observe that
$$ \sum_{(x_1, \dots, x_{M}) \in \left(\{-1, 1\}^{N}\right)^M} \zeta(x_1, \dots, x_{M}) \text{AND-OR}^{M}_{N}(x_1, \dots, x_{M})$$
$$= 2^{M} \mathbf{E}_\mu [\Psi( \dots, \sgn(\psi(x_i)), \dots) \text{AND}_M\left( \dots, \text{OR}_N(x_i), \dots\right)]$$
\begin{equation} \label{eq1andor} = \sum_{z \in \{-1, 1\}^{M}} \Psi(z) \left(\sum_{(x_1, \dots, x_{M}) \in \left(\{-1, 1\}^N\right)^M} \text{AND}_M\left( \dots,  \text{OR}_N(x_i), \dots\right) \mu(x_1, \dots, x_{M}|z)\right),\end{equation}
where $\mu(\mathbf{x}|z)$ denotes the probability of $\mathbf{x}$ under $\mu$, conditioned on $(\dots, \sgn(\psi(x_i)), \dots)=z$.

Let $A_{1} = \{x \in \{-1, 1\}^{N}: \psi(x) \geq 0, \text{OR}_N(x) = -1\}$ and $A_{-1} = \{x \in \{-1, 1\}^{N}: \psi(x) < 0, \text{OR}_N(x) = 1\}$, so
$A_1 \cup A_{-1}$ is the set of all inputs $x$ where the sign of $\psi(x)$ disagrees with $\text{OR}_N(x)$. Notice that 
$\sum_{x \in A_1 \cup A_{-1}} |\psi(x)| = \delta/2$ because $\psi$ has correlation $1-\delta$ with $\text{OR}_N$. 

As noted in \cite{sherstovFOCS}, for any given $z \in \{-1, 1\}^{M}$, the following two random variables are identically distributed:

\begin{itemize}
\item The string $(\dots, \text{OR}_N(x_i), \dots)$ when one chooses $(\dots, x_i, \dots)$ from the conditional distribution 
$\mu(\cdot|z)$. 
\item The string $(\dots, y_iz_i, \dots)$, where $y \in \{-1, 1\}^{M}$ is a random string whose $i$th bit independently
takes on value $-1$ with probability $2 \sum_{x \in A_{z_i}} |\psi(x)| \le \delta$. 
\end{itemize}

Thus, Expression (\ref{eq1andor}) equals

\begin{equation} \label{eq2andor} \sum_{z \in \{-1, 1\}^{M}} \Psi(z) \cdot \mathbf{E}[\text{AND}_M(\dots, y_iz_i, \dots)],\end{equation}

where $y \in \{-1, 1\}^{M}$ is a random string whose $i$th bit independently
takes on value $-1$ with probability $2 \sum_{x \in A_{z_i}} |\psi(x)| \le \delta$. 

We first argue that the term corresponding to $z=-\mathbf{1}_{M}$ contributes $-\Psi(z)$ to Expression (\ref{eq2andor}).
By \eqref{eq:facts1} of \lemref{facts}, if $\text{OR}_N(x) = 1$ (i.e., if $x=\mathbf{1}_{N}$), then $\sgn(\psi(x))=1$. This implies that $A_{-1}$ is empty; that is, if $\sgn(\psi(x))=-1$,
then it must be the case that $\text{OR}_N(x)=-1$. Therefore, for $z=-\mathbf{1}_{M}$, the $y_i$'s are all $-1$ with probability 1, and hence $\mathbf{E}_y[\text{AND}_M\left(\dots, y_iz_i, \dots\right)] = \text{AND}_M\left(-\mathbf{1}_{M}\right) = -1$.
Thus the term corresponding to $z=-\mathbf{1}_{M}$ contributes $-\Psi(z)$ to Expression (\ref{eq2andor}) as claimed.

All $z \neq -\mathbf{1}_{M}$ can be handled as in Sherstov's proof of Proposition \ref{prop:sherstov}, because $\text{AND}_{M}$ has low block sensitivity
at these inputs. To formalize this, we invoke the following proposition, whose proof we provide in Appendix \ref{app:andor1} for completeness.

\begin{proposition}[\cite{sherstovFOCS}] \label{finalprop} Let $F: \{-1, 1\}^{M} \rightarrow \{-1, 1\}$ be a given Boolean function.
Let $y \in \{-1, 1\}^{M}$ be a random string whose $i$th bit is set to $-1$ with probability at most $\alpha \in [0, 1]$, 
and to $+1$ otherwise, independently for each $i$. Then for every $z \in \{-1, 1\}^{M}$, 
$$\mathbf{P}_y[F(z_1, \dots, z_{M}) \ne F(z_1y_1, \dots, z_My_M)] \leq 2\alpha\operatorname{bs}_z(F).$$ 
\end{proposition}

In particular, since $\text{bs}_z(\text{AND}_M) = 1$ for all $z \neq - \mathbf{1}_M$, Proposition \ref{finalprop} implies that for all $z \neq -\mathbf{1}_{M}$,
and $F=\text{AND}_{M}$, $\mathbf{P}_y[F(z_1, \dots, z_{M})= F(z_1y_1, \dots, z_ky_k)] \geq 1-2\delta$.

Recalling that the term corresponding to $z = -\mathbf{1}_M$ contributes $-\Psi(-\mathbf{1}_M)$ to the sum, we obtain the following lower bound on Expression (\ref{eq2andor}).

$$\sum_{z \in \{-1, 1\}^{M}} \Psi(z) \cdot \mathbf{E}[\text{AND}_M\left(\dots, y_iz_i, \dots\right)] \geq 
-\Psi(-\mathbf{1}_{M}) + \left(\sum_{z \neq -\mathbf{1}_M} \Psi(z) \text{AND}_M(z)\right) - 4\delta \left( \sum_{z \neq -\mathbf{1}_M} |\Psi(z)| \right) $$
$$ \ge \left(\sum_{z \in \{-1, 1\}^M} \Psi(z) \text{AND}_M(z)\right) -4\delta = \eps - 4\delta>1/3.$$ \end{proof}

This completes the proof of Theorem \ref{thm:andor}.
\end{proof}

\begin{remark} \v{S}palek \cite{spalek} has exhibited an explicit dual witness showing that the $\eps$-approximate degree of both the $\text{AND}$ function and the $\text{OR}$ function is $\Omega(\sqrt{n})$, for $\eps = 1/14$ (in fact, we generalize \v{S}palek's construction to any symmetric function in Section \ref{secfour} below). In Section \ref{app:spalek} we show how to generalize \v{S}palek's argument in a different way to handle any constant $\eps \in (0, 1)$. With these dual polynomials in hand, the dual solution $\zeta$ given in our proof
is completely explicit. This answers a question of \v{S}palek \cite[Section 4]{spalek} in the affirmative.
\end{remark}

\section{Dual Polynomials for Symmetric Boolean Functions}
\label{secfour}
In this section, we construct a dual polynomial witnessing a tight lower bound on the approximate degree of any symmetric function. The lower bound we recover was first proved by Paturi \cite{paturi} via a symmetrization argument combined with the classical Markov-Bernstein inequality from approximation theory (see Section \ref{sec:markovbernstein}). Paturi also provided a matching upper bound. \v{S}palek \cite{spalek}, building on work of Szegedy, exhibited an explicit dual witness to the $\Omega(\sqrt{n})$ approximate degree of the $\operatorname{OR}$ function and asked whether one could construct an analogous dual polynomial for the symmetric $t$-threshold function \cite[Section 4]{spalek}. We accomplish this in the more general case of arbitrary symmetric functions by extending the ideas underlying \v{S}palek's dual polynomial for $\text{OR}$.

\subsection{Symmetric functions}

For a vector $x \in \{-1, 1\}^n$, let $|x| = \frac{1}{2}(n - (x_1 + \dots + x_n))$ denote the number of $-1$'s in $x$. A Boolean function $f : \{-1, 1\}^n \to \{-1, 1\}$ is symmetric if $f(x) = f(y)$ whenever $|x| = |y|$. That is, the value of $f$ depends only on the number of inputs that are set to $-1$. The simplest symmetric functions are the $t$-threshold functions:
\[\tau_t(x) = \begin{cases}
-1 & \text{if } |x| \ge t \\
1 & \text{otherwise}.
\end{cases}\]
Important special cases include $\operatorname{OR} = \tau_1$, $\operatorname{AND} = \tau_n$, and the majority function $\operatorname{MAJ} = \tau_{\lceil n/2 \rceil}$. Let $[n] = \{0, 1, \dots, n\}$. To each symmetric function $f$, we can associate a unique univariate function $F: [n] \to \{-1, 1\}$ by taking $F(|x|) = f(x)$. Throughout this section, we follow the convention that lower case letters refer 
to multivariate functions, while upper case letters refer to their univariate counterparts.

We now discuss the dual characterization of approximate degree established in  \thmref{thm:prelim} as it applies to symmetric functions. Following the notation in \cite{spalek}, the standard inner product $p \cdot q = \sum_{x \in \{-1, 1\}^n} p(x)q(x)$ on symmetric functions $p, q$ induces an inner product on the associated univariate functions:
\[P \cdot Q := \sum_{i = 0}^n {n \choose i} P(i) Q(i).\]
We refer to this as the \emph{correlation} between $P$ and $Q$. Similarly, the $\ell_1$-norm $\|p\|_1 = \sum_{x \in \{-1, 1\}^n} |p(x)|$ induces a norm $\|P\|_1 = \sum_{i = 0}^n {n \choose i} P(i)$. These definitions carry over verbatim when $f$ is real-valued instead of Boolean-valued.

If $f$ is symmetric, we can restrict our attention to symmetric $\phi$ in the statement of \thmref{thm:prelim} and it becomes convenient to work with the following reformulation of \thmref{thm:prelim}.

\begin{corollary}\label{cor:dual}
A symmetric function $f:\{-1, 1\}^n \to \{-1, 1\}$ has $\eps$-approximate degree greater than $d$ if and only if there exists a symmetric function $\phi: \{-1, 1\}^n \to \mathbb{R}$ with pure high degree $d$ such that
\[\frac{\Phi \cdot F}{\|\Phi\|_1} = \frac{\phi \cdot f}{\|\phi\|_1} > \eps.\]
(Here, $F$ and $\Phi$ are the univariate function associated to $f$ and $\phi$, respectively).
\end{corollary}

We clarify that the pure high degree of a multivariate polynomial $\phi$ does not correspond to the smallest degree of a monomial in the associated univariate function $\Phi$ (even though the ordinary degree of a symmetric $\phi$ is the largest degree of a monomial in $\Phi$). When we talk about the pure high degree of a univariate polynomial $\Phi$, we mean the pure high degree of its corresponding multilinear polynomial $\phi$.

We exploit the following method for constructing polynomials of pure high degree $d$. Let $\psi$ be a multivariate polynomial of degree $n-d$, and let $\chi_{[n]}(x)$ denote the parity function
on $n$ variables. Consider the function $\phi(x) = \psi(x)\chi_{[n]}(x)$, i.e., $\phi$ is obtained by multiplying $\psi$ by the parity function. It is straightforward to check that $\phi$ has pure high degree $d$. 
Notice that if $\psi$ is symmetric, then so is $\phi$, and the corresponding univariate polynomials satisfy $\Phi(k) = \Psi(k) \cdot (-1)^k$.
Therefore, to show that a symmetric function $f$ with a ``jump'' at $t$ has approximate degree greater than $d$, it is enough to exhibit an $(n-d)$-degree univariate polynomial $\Psi$ such that $(-1)^i \Psi(i)$ has high correlation with its associated univariate function $F$.

We are now in a position to state the lower bound that we will prove in this section. Paturi \cite{paturi} completely characterized the approximate degree of a symmetric Boolean function by the location of the layer $t$ closest to the center of the Boolean hypercube such that $F(t-1) \ne F(t)$. 

\begin{theorem} [{\cite[Theorem 4]{paturi}}]
Given a nonconstant symmetric Boolean function $f$ with associated univariate function $F$, let $\Gamma(f) = \min \{|2t - n - 1| : F(t-1) \ne F(t), 1 \le k \le n\}$. Then $\widetilde{\deg}(f) = \Theta(\sqrt{n(n - \Gamma(f)})$.
\end{theorem}

Paturi proved the upper bound non-explicitly by appealing to the Jackson theorems from approximation theory. He proved the lower bound by combining symmetrization with an appeal to the Markov-Bernstein inequality (see Section \ref{sec:markovbernstein}) -- however, his proof does not yield an explicit 
dual polynomial. We construct an explicit dual polynomial to prove the following proposition, which is easily seen to imply Paturi's lower bound.

\begin{proposition} \label{prop:sym-lb}
Given $f$ and $F$ as above, let $1 \le t \le n$ be an integer with $F(t-1) \ne F(t)$. Then $\widetilde{\deg}(f) = \Omega(\sqrt{t(n-t+1)})$.
\end{proposition}

In particular, the approximate degree of the symmetric $t$-threshold function is $\Omega(\sqrt{t(n-t+1)})$. This special case serves as a useful model for understanding our construction.

\subsection{Proof outline}
 \label{sec:sym-outline}
We start with an intuitive discussion of \v{S}palek's construction of a dual polynomial for $\operatorname{OR}$, with the goal of elucidating how we extend the construction to arbitrary symmetric functions. Consider the perfect squares $S = \{k^2 : k^2 \le n\}$ and the univariate polynomial
\[R(x) = \frac{1}{n!}\prod_{i \in [n] \setminus S} (x - i).\]
This polynomial is supported on $S$, and for all $k^2 \in S$,
\[{n \choose k^2} |R(k^2)| = {n \choose k^2} \cdot \frac{1}{n!} \cdot \frac{\prod_{\substack{i \in [n] \\ i \ne k^2}}|k^2 - i|}{\prod_{\substack{i \in S \\ i \ne k^2}} |k^2 - i|} = \frac{1}{\prod_{\substack{i \in S \\ i \ne k^2}} |k^2 - i|}.\]
Note the remarkable cancellation in the final equality. This quotient is maximized at $k = 1$. In other words, the threshold point $t = 1$ makes the largest contribution to the $\ell_1$ mass of $R$. Moreover, one can check that $R(0)$ is only a constant factor smaller than $R(1)$. 

\v{S}palek exploits this distribution of the $\ell_1$ mass by considering the polynomial $P(x) = R(x) / (x - 2)$. The values of $P(x)$ are related to $R(x)$ by a constant multiple for $x = 0, 1$, but $P(k)$ decays as $|P(k^2)| \approx |R(k^2)|/k^2$ for larger values. This decay is fast enough that a \emph{constant fraction} of the $\ell_1$ mass of $P$ comes from the point $P(0)$.\footnote{It is also necessary to check that $P(2)$ is only a constant factor larger than $P(0)$.} Now $P$ is an $(n-\Omega(\sqrt{n}))$-degree univariate polynomial, so we just need to show that $Q(i) = (-1)^iP(i)$ has high correlation with $\operatorname{OR}$. We can write
\[Q \cdot \operatorname{OR} = 2Q(0) - Q \cdot \mathbf{1} = 2Q(0),\]
since the multilinear polynomial associated to $Q$ has pure high degree $\Omega(\sqrt{n})$, and therefore has zero correlation with constant functions. Because a constant fraction of the $\ell_1$ mass of $Q$ comes from $Q(0)$, it follows that $|Q \cdot \operatorname{OR}| / \|Q\|_1$ is bounded below by a constant. By perhaps changing the sign of $Q$, we get a good dual polynomial for $\operatorname{OR}$.

A natural approach to extend \v{S}palek's argument to symmetric functions with a ``jump'' at $t$ is the following:
\begin{enumerate}[{Step}~1:]
\item Find a set $S$ with $|S| = \Omega(\sqrt{t(n-t+1)})$ such that the maximum contribution to the $\ell_1$ norm of $R(x) = \frac{1}{n!}\prod_{i \in [n] \setminus S} (x - i)$ comes from the point $x = t$. Equivalently,
\[{n \choose j} |R(j)| = \frac{1}{\prod_{\substack{i \in S \\ i \ne j}} |j - i|}\]
is maximized at $j = t$.
\item Define a polynomial $P(x) = R(x) / (x-(t-1))(x-(t+1))$. Dividing $R(x)$ by the factor $(x-t-1)$ is analogous to \v{S}palek's division of $R(x)$ by $(x-2)$.
We also divide by $(x-t+1)$ because we will ultimately need our polynomial $P(x)$ to decay faster than \v{S}palek's by a factor of $|x-t|$
as $x$ moves away from the threshold. By dividing by both $(x-t-1)$ and $(x-t+1)$, we ensure that most of the $\ell_1$ mass of $P$ is concentrated at the points $t-1, t, t+1$.
\item Obtain $Q$ by multiplying $P$ by parity, and observe that $Q(t-1)$ and $Q(t)$ have opposite signs. Since $F(t-1)$ and $F(t)$ also have opposite signs, we can ensure that both $t-1$ and $t$ contribute positive correlation. Suppose these two points contribute a $1/2 + \eps$ constant fraction of the $\ell_1$-norm of $Q$. Then even in the worst case where the remaining points all contribute negative correlation, $Q \cdot F$ is still at least a $2\eps$ fraction of $\|Q\|_1$ and we have a good dual polynomial. Notice that the pure high degree of $Q$ is $|S| + 2$, yielding the desired lower bound.
\end{enumerate}

In Section \ref{sec:nextsec}, we carry out this line of attack in the case where $t = \Omega(n)$. This partial result also gives the right intuition for general $t$, although the details are somewhat more complicated. Namely, in Step 3, we may need to rely on the alternative points $t$ and $t + 2$ to contribute high positive correlation between $F$ and $Q$, rather than inputs $t-1$ and $t$.

\subsection{A Dual Polynomial for $\text{MAJ}$}
\label{sec:nextsec}
We first construct a dual polynomial that witnesses an $\Omega(t)$ lower bound for symmetric functions having a ``jump'' at $t \le n/2$. Notice that this bound matches \propref{prop:sym-lb} if $t = \Omega(n)$, but is weaker otherwise (e.g. in the case of OR). By setting $t = \lceil \frac{n}{2} \rceil$, we can write down down a clean dual polynomial for the majority function $\text{MAJ}$. This case is illustrative, as one
can view \v{S}palek's dual polynomial for $\text{OR}$ and our dual polynomial for $\text{MAJ}$ as two ends of a spectrum, with our general construction interpolating between the two extremes.

\begin{proposition} \label{prop:maj-lb}
Let $f: \{-1, 1\}^n \to \{-1, 1\}$ be a Boolean function with associated univariate function $F$. If $1 \le t \le n/2$ such that $F(t-1) \ne F(t)$, then $\widetilde{\deg}(f) = \Omega(t)$.
\end{proposition}

\begin{proof}
We follow the proof outline given in the previous section. Define the set
\[S = \{t \pm 4 \ell : 0 \le \ell \le t/4\}.\]
Note that $|S| = \Omega(t)$. We claim that $\pi_S(i) := \prod_{j \in S, j \ne i} |j - i|$ is minimized at $i = t$. Notice that translating all points in $S$ by a constant does not affect $\pi_S(i)$, and scaling all points in $S$ 
by a constant does not affect $\text{argmin}_i \pi_S(i)$. Thus, it is enough to show that $\pi_{S^*}(i)$ is minimized at $i=0$ for the set $S^* = \{\pm \ell : \ell \le t\}$. In this case, $\pi_{S^*}(i)$ takes the simple form $(t - i)!(t + i)!$, 
and we see that
\[\frac{\pi_{S^*}(0)}{\pi_{S^*}(i)} = \frac{(t!)^2}{(t-i)!(t+i)!} = \frac{t}{t+|i|} \cdot \frac{t-1}{t+|i|-1} \cdot \dots \cdot \frac{t - |i| + 1}{t + 1}\]
is a product of terms smaller than $1$, so $\pi_{S^*}(i)$ is indeed minimized at $i = 0$.

With Step 1 completed, we let $T = S \cup \{t - 1, t + 1\}$ and define the polynomial
\[P(x) = (-1)^s \frac{4^{2h}(h!)^2}{n!} \prod _{j \in [n] \setminus T} (x - j),\]
where $h = \lfloor t/4 \rfloor$ and $s$ is a sign bit to be determined later. The normalization is chosen so that ${n \choose t} |P(t)| = 1$. We divide by both $(x - (t-1))$ and $(x - (t+1))$ to ensure that the rate of 
decay of $P(x)$ is at least quadratic as $x$ moves away from $t$. This will ultimately allow us to show that most of the $\ell_1$ mass of $P$ comes from the points $x = t-1$ and $x = t$.

Write the $\ell_1$ contribution due to the point $r$ as
\[{n \choose r}|P(r)| = {n \choose r} \frac{4^{2h}(h!)^2}{n!} \frac{\prod_{j \in [n] \setminus \{r\}} |r - j|}{\prod_{j\in T \setminus \{r\}} |r - j|} = \frac{4^{2h}(h!)^2}{\prod_{j \in T \setminus \{r\}} |r - j|}.\]
For $r = t \pm 1$ this becomes
\begin{align*}
\frac{4^{2h}(h!)^2}{2 \prod_{\ell = 1}^h (4\ell - 1)(4\ell + 1)} &= \frac{1}{2} \prod_{\ell = 1}^h \left(1 + \frac{1}{16\ell^2 - 1}\right) \\
&\le \frac{1}{2} \exp\left(\sum_{\ell = 1}^h \frac{1}{15\ell^2}\right) \\
&\le \frac{1}{2} e^{\pi^2 / 90} < 1,
\end{align*}
where the first inequality holds because $1+x \leq e^x$ for all $x \geq 0$.
This shows that the $\ell_1$ contributions of the points $t-1$ and $t+1$ are equal, and not too large:
\[{n \choose t - 1} |P(t-1)| = {n \choose t + 1} |P(t+1)| < 1.\]

Now we analyze the remaining summands, and show that their total contribution is much smaller than $1$. Recall that the choice $i = t$ minimizes $\pi_S(i)$, and that $\pi_S(t)=4^{2h}(h!)^2$. Therefore,
\[{n \choose t + 4\ell} |P(t + 4\ell)| = \frac{4^{2h}(h!)^2}{\prod_{j \in T \setminus \{t+4\ell\}} |t + 4\ell - j|} \le \frac{1}{|4\ell + 1||4\ell-1|} \le \frac{1}{15\ell^2}.\]
We can use this quadratic decay to bound the total $\ell_1$ mass of the points outside of $\{t-1, t, t+1\}$:
\[\sum_{j \in S \setminus \{t\}} {n \choose j} |P(j)| \le \sum_{\ell=-h}^{h} \frac{1}{15\ell^2} \le \frac{2}{15} \cdot \frac{\pi^2}{6} < \frac{1}{4}.\]

For the final part of our construction, we multiply $P$ by parity to get $Q(i) = (-1)^iP(i)$. Since $P(t-1)$ and $P(t)$ have the same sign, $Q(t-1)$ and $Q(t)$ have opposite signs. Since $F(t-1)$ and $F(t)$ also have opposite signs,
we can choose $s\in \{-1, 1\}$ to ensure that 
\begin{align*}
Q \cdot F &> {n \choose t - 1} |P(t-1)| + {n \choose t} |P(t)| - {n \choose t + 1} |P(t+1)| - \sum_{j \in S \setminus\{t\}} {n \choose j} |P(j)| \\
& \ge 1 - \frac{1}{4} = \frac{3}{4}.
\end{align*}
As the total $\ell_1$ mass $\|P\|_1$ is at most $3 + \frac{1}{4}$, we get that $(Q \cdot F) / \|Q\|_1 > \frac{3}{13}$. By \corref{cor:dual}, the $\frac{3}{13}$-approximate degree of $f$ is $\Omega(t)$.
\end{proof}

\subsection{General Symmetric Boolean Functions}
\label{sec:generalsymmfuncs}
We now show how to generalize our dual polynomial for MAJ and \v{S}palek's dual polynomial for OR to handle arbitrary symmetric functions. Recall that we are given a Boolean function $f$, an associated univariate function $F$, and a number $t$ such that $F(t-1) \ne F(t)$. As our goal is to show that $\widetilde{\deg}(F) = \Omega(\sqrt{t(n-t+1)})$, we may without loss of generality assume that $t \le n/2$ throughout this section. As the case of $t = 1$ is handled by \v{S}palek's construction, we may also assume $t \ge 2$ to improve the constants in our analysis.

As a first attempt at defining a suitable set $S$ for use in constructing a dual polynomial for $f$, we consider the set $S' = \{tk^2 : k^2 \le n/t\}$. Fact \ref{fact:min-prod} in \appref{app:sym-dual} implies that $\prod_{i \in S', i \ne j} |j - i|$ is minimized at $t$. Unfortunately, the set $S'$ is too small -- it has size only $\Theta(\sqrt{n/t})$ instead of $\Theta\left(\sqrt{t(n-t+1)}\right)$. The trick is to notice that the distance between any two points in $S'$ is at least $t$. Therefore, we should be able to interlace $\Theta(t)$ translated copies of $S'$, and still have the desired product minimized near $t$. The following lemma gives the details, but its proof is rather technical and deferred to \appref{app:sym-dual}.

\begin{lemma}\label{lemma:min-block}
Let 
\[S = \{tk^2 + 4\ell : 1 \le k \le \sqrt{\left(n-t+1\right)/t}, 0 \le \ell \le ct\} \cup \{t - 4\ell : 0 \le \ell \le ct\}\]
where $c \le 1/32$. Then for $i \in S$, the product
\[\pi_S(i) := \prod_{\substack{i' \in S \\ i' \ne i}} |i - i'|\]
is minimized for some $i^*$ with $(1-4c)t \le i^* \le (1 + 4c)t$.
\end{lemma}

Thus the product of differences is minimized somewhere in a $\Theta(t)$-sized neighborhood of $t$. The exact location depends delicately on $n$ and $t$. However, since the product $\pi_S(i)$ is invariant under translations of $S$, we can assume that the minimizer $i^*$ is one of the points $t-1, t,$ or $t + 1$.

For intuition, observe that one can view the set $S$ of \lemref{lemma:min-block} as interpolating between the set for OR used by \v{S}palek, and the set used to prove our $\Omega(t)$ lower bound in Proposition \ref{prop:maj-lb}. Notice that $S$ contains all points of the form $t \pm 4\ell$, plus additional points corresponding to perfect squares when $t = o(n)$.

Let $S$ be the set in \lemref{lemma:min-block} (or a translate thereof), and suppose the corresponding product is minimized at $i^*$. Let $T = S \cup \{i^* - 1, i^* + 1\}$. Define the polynomial
\[P(x) = (-1)^s\frac{\pi_S(i^*)}{n!}\prod_{j \in [n] \setminus T} (x - j),\]
where the sign bit $s$ is to be determined later. The choice of normalization is so that ${n \choose i^*} |P(i^*)| = 1$. Our goal is now to show that the $\ell_1$ mass of $P$ is concentrated at the points $i^* - 1, i^*$ and $i^* + 1$. The following lemma shows that the contribution of a point $x$ to the $\ell_1$ mass of $P$ decays at a quadratic rate as $x$ moves away from $i^*$. This is precisely because we include both points $i^* - 1$ and $i^* + 1$ in $T$. Full details are in \appref{app:sym-dual}.

\begin{lemma}\label{lemma:term-bound}
Let $r = tk^2 + 4\ell$ for $k \ge 2$ and $0 \le \ell \le ct$. Then ${n \choose r}|P(r)| \le 1/(t^2(k^2 - 2)^2)$. If $v = i^* + 4\ell$, then ${n \choose v}|P(v)| \le 1 / (16\ell^2 - 1)$.
\end{lemma}

Since the sum of the inverse squares of the integers is bounded by a constant, the total contribution of these points to $\|P\|_1$ is dominated by the mass contributed by $P(i^*)$.

\begin{lemma}\label{lemma:tail-bound}
\[\sum_{j \in T \setminus \{i^*-1, i^*, i^*+1\}} {n \choose j} |P(j)| \le \frac{2}{5}.\]
\end{lemma}

This bound allows us to sketch a proof of \propref{prop:sym-lb} in full generality.
\medskip

\begin{proof}[Proof sketch of \propref{prop:sym-lb}:]
We consider three cases based on which of ${n \choose i^*-1} |P(i^*-1)|, {n \choose i^*} |P(i^*)| = 1, $ and ${n \choose i^*+1} |P(i^*+1)|$ is the smallest. The relationship between these terms determines how we choose the location of $i^*$ relative to the ``jump'' at $t$. We set $i^*$ so that after multiplying $P$ by parity to obtain a polynomial $Q$, the larger two of these terms contribute positively to $Q \cdot F$. They will hence dominate the (possibly negative) correlation due to the smallest term, as well as the contribution of size at most $2/5$ due to remaining points in $T$. Ultimately, we show that $(Q \cdot F) / \|Q\|_1 \le \frac{1}{14}$, which gives the asserted lower bound by \corref{cor:dual}. The calculations for each of these cases are analogous to those in the proof of \propref{prop:maj-lb} and given in \appref{app:sym-dual}.

\end{proof}

\subsection{On Complementary Slackness}
In this section, we give some additional intuition, based
on complementary slackness,
that helps to explain the structure of the dual polynomial 
exhibited in Section \ref{sec:generalsymmfuncs}. 
The discussion that follows is deliberately informal
and is meant to complement the formal argument given
in Section \ref{sec:generalsymmfuncs}.

We illustrate the idea by considering the symmetric $t$-threshold function $f$, which evaluates to $-1$ on
inputs of Hamming weight at least $t$ and evaluates
to $1$ on all other inputs.
In \cite{sherstov:incexc}, Sherstov gives an explicit, asymptotically optimal polynomial $p$ for approximating
$f$ in the $\ell_\infty$ norm. 
If this polynomial $p$ were in fact an \emph{exactly} optimal
solution to the primal linear program of Section \ref{sec:dualcharacterization}, then 
complementary slackness (cf. \cite[pg. 95]{schrijver})
would imply that the optimal dual polynomial $\phi$ is supported on the points corresponding to the constraints made tight by the primal optimal polynomial $p$.
That is, it would hold that $\phi(x)=0$ except for those $x\in \{-1, 1\}^n$ for which $|p(x)-f(x)| = \eps$. We will refer to such values of $x$ as \emph{maximum-error points of $p$.} 

While it is not clear whether Sherstov's polynomial $p$ is exactly optimal, our dual polynomial is still approximately consistent with the conditions obtained by applying complementary slackness to $p$. Sherstov's construction of $p$ works by taking a Chebyshev polynomial of degree $\Theta(\sqrt{n/t})$, shifting and scaling it, and then composing it with a Chebyshev polynomial of degree $\Theta(t)$. This is reminiscent of our dual solution, which interlaces $\Theta(t)$ copies of a set of size $\Theta(\sqrt{n/t})$. In general, it is difficult to determine the precise maximum-error points of Sherstov's polynomial $p$. However, our dual polynomial can be viewed as placing nonzero weight on close \emph{approximations} to the maximum-error points of $p$. We explain this viewpoint below.

Let $T_d: \mathbb{R} \rightarrow \R$ denote the degree-$d$
Chebyshev polynomial of the first kind. It is well-known
that the extreme points of $T_d$ are the degree-$d$ \emph{Chebyshev nodes},
which take the form $\cos(k \pi/d)$ for $0 \leq k \leq d$. 
Truncating the Taylor expansion of $\cos(x)=1-x^2/2+\dots$ after the quadratic term, one sees that for $d=\sqrt{n}$, 
$\cos(k \pi/d) \approx 1-(ck^2/d^2)=1-ck^2/n$ for some constant $c$. 

It is known \cite{NS94} that an appropriately shifted-and-scaled Chebyshev polynomial $Q_d$ of degree $d = \Theta(\sqrt{n})$ itself yields an asymptotically optimal approximation $Q_d(\sum_{i=1}^n x_i / n)$ to the OR function. Recall from Section
\ref{sec:sym-outline}
that \v{S}palek's dual polynomial for the OR function \cite{spalek}
only places nonzero weight on inputs of Hamming weight
equal to a perfect square (or equal to two). We can therefore view \v{S}palek's dual polynomial for the OR function as placing
nonzero weight only on points whose Hamming weight closely approximates 
a constant multiple of a Chebyshev node 
(Section \ref{app:spalek} shows that 
the dual polynomial for the OR function is
robust to scaling the non-vanishing Hamming weight values by constants, i.e., it suffices to place nonzero weight only on inputs
of Hamming weight $ck^2$ or $1$ for any constant $c\ge 2$).

Moving to the general case and eliding many details, Sherstov approximates the symmetric $t$-threshold function $f$ with a polynomial $p$ roughly of the form
$$p(x) = T_{t}\left(q\left(\sum_{i=1}^n x_i/n\right)\right),$$
where
$q$ is a shifted-and-scaled version of the Chebyshev polynomial
of degree $\sqrt{n/t}$.
Following the intuition above that the degree-$d$ Chebyshev nodes are approximated by points of the form $1-ck^2/d^2$ for some constant $c$, the inner polynomial $q(\sum_{i=1}^n x_i/n)$ in Sherstov's construction hits its extreme points at inputs of Hamming weight close to $ctk^2$ for the non-negative integers $k \leq \sqrt{n/t}$. Moreover, $q$ alternates between $-1$ and $+1$ at its extreme points. Thus, as the Hamming weight of the input $x$ increases from $ctk^2$ to $ct(k+1)^2$, the inner polynomial $q(\sum_{i=1}^n x_i/n)$ passes through all $t$ maximum-error points of the \emph{outer} polynomial $T_t$.

Note also that the degree-$d$ Chebyshev nodes $\cos(k\pi/d) \approx 1-k^2/d^2$ are clustered near the endpoints of the interval $[-1, 1]$ rather than the middle of the interval, so most of the maximum-error points
of the composed polynomial in fact fall very close to inputs of Hamming weight
$ctk^2$.

To see how these maximum-error points correspond to the support of our dual polynomial $\phi$ for the $t$-threshold function, recall that there is some constant $c'$ such $\phi$ takes nonzero values only on inputs with Hamming weight in the set
\[S = \{tk^2 + 4\ell : 1 \le k \le \sqrt{\left(n-t+1\right)/t}, 0 \le \ell \le c't\} \cup \{t - 4\ell : 0 \le \ell \le c't\}.\]
Roughly speaking, our dual witness thus takes nonzero values only on inputs of Hamming weight
very close to $tk^2$ for each $k \leq \sqrt{n/t}$ (i.e., for each Hamming weight of the form
$tk^2$, our dual witness takes nonzero values on $t$ distinct Hamming weights in the vicinity of $tk^2$), just as as predicted above.

\subsection{A Dual Polynomial for the $\eps$-Approximate Degree of OR} \label{app:spalek}

\v{S}palek \cite{spalek} constructed an explicit dual witness for the fact that the OR function on $n$ variables has $(1/14)$-approximate degree $\Omega(\sqrt{n})$. We extend his argument to exhibit a dual witness that shows that OR has $\eps$-approximate degree $\Omega(\sqrt{n})$ for \emph{any} constant $\eps \in (0, 1)$.

\begin{proposition} \label{prop:maj-lb}
Let $\eps \in (0, 1)$. Then $\operatorname{OR}_n$ has approximate degree $\Omega(\sqrt{n (1-\eps)})$.
\end{proposition}

\begin{proof}
As before, we associate with each symmetric function
$p$ a univariate function $P$ and vice versa. Let $c = \lceil 8 / (1-\eps)\rceil$. Let $m = \lfloor \sqrt{n/c} \rfloor$ and define the set
\[T = \{1\} \cup \{ck^2 : 0 \le k \le m\}.\]
Note that $|T| = \Omega(\sqrt{n/c})$. Define the polynomial
\[P(x) = (-1)^s \frac{c^{2m}(m!)^2}{n!} \prod_{j \in [n] \setminus T} (x - j),\]
where $s$ is a sign bit to be determined later. It is easy to check that $|P(0)| = 1$.

The $\ell_1$ contribution due to the $r$'th layer of the Boolean hypercube  is
\[{n \choose r}|P(r)| = {n \choose r} \frac{c^m(m!)^2}{n!} \frac{\prod_{j \in [n] \setminus \{r\}} |r - j|}{\prod_{j\in T \setminus \{r\}} |r - j|} = \frac{c^m(m!)^2}{\prod_{j \in T \setminus \{r\}} |r - j|}.\]

For $r=1$ the right hand side evaluates to
\[\frac{c^m(m!)^2}{\prod_{i = 1}^m (ci^2 - 1)} = \prod_{i=1}^m \frac{i^2}{i^2 - 1/c} \ge 1.\]
Thus, the total $\ell_1$ contribution of the inputs of Hamming weight 1 is at least $1$. 

For $r = ck^2$ where $k > 0$, we get
\begin{align*}
\frac{c^m(m!)^2}{(ck^2 - 1)\prod_{i \in [m] \setminus \{k\}} |ci^2 - ck^2|} &= \frac{(m!)^2}{(ck^2 - 1)\prod_{i \in [m] \setminus \{k\}} (i+k)|i-k|} \\
&= \frac{2(m!)^2}{(ck^2 - 1)(m+k)!(m-k)!} \\
&\le \frac{2}{ck^2-1}
\end{align*}
where the last inequality follows as in \cite{spalek} because
\[\frac{(m!)^2}{(m+k)!(m-k)!} = \frac{m}{m+k} \cdot \frac{m-1}{m+k-1} \cdot \ldots \cdot \frac{m-k+1}{m+1}\]
is a product of factors that are each smaller than 1. This shows that the total $\ell_1$ contribution of the Hamming layers excluding 0 and 1 is at most
\[\sum_{k=1}^{m} \frac{2}{ck^2 - 1} < \sum_{k=1}^\infty \frac{4}{ck^2} < \frac{8}{c}.\]

For the final part of our construction, we let $Q(i) = (-1)^iP(i)$. Then the multilinear polynomial corresponding to $Q$ has pure high degree $\Omega(\sqrt{n(1-\eps)})$. Since $P(0)$ and $P(1)$ have the same sign, $Q(0)$ and $Q(1)$ have opposite signs. Since $\operatorname{OR}(0)$ and $\operatorname{OR}(1)$ also have opposite signs,
we can choose $s \in \{-1, 1\}$ to ensure that 
\begin{align*}
Q \cdot \operatorname{OR} &\ge {n \choose 0} |P(0)| + {n \choose 1} |P(1)| - \sum_{j \in S \setminus\{0, 1\}} {n \choose j} |P(j)| \\
& \ge 1 + n|P(1)| - \frac{8}{c}.
\end{align*}
As the total $\ell_1$ mass $\|Q\|_1$ of $Q$ is at most $1 + n|P(1)| + 8/c$, we see that 
\[\frac{Q \cdot \operatorname{OR}}{\|Q\|_1} \ge \frac{1 + n|P(1)| - 8/c}{1 + n|P(1)| + 8/c} = 1 - \frac{16}{c + cn|P(1)| + 8} \ge 1 - \frac{16}{2c + 8}.\]
Since $c > 8/(1-\eps)$, the right hand side is at least $\eps$. \corref{cor:dual} then implies that the $\eps$-approximate degree of OR is $\Omega(\sqrt{n(1-\eps)})$.
\end{proof}

\section{A Constructive Proof of Markov-Bernstein Inequalities}
\label{sec:markovbernstein}
The Markov-Bernstein inequality for polynomials with real coefficients asserts that
$$|p'(x)| \leq \min\left\{ \frac{n}{\sqrt{1-x^2}} ,n^2\right\} \|p\|_{[-1,1]}, x \in (-1,1)$$
for every real polynomial of degree at most $n$. Here, and in what follows, $$\|p\|_{[-1, 1]} := \sup_{y \in [-1, 1]} |p(y)|.$$

This inequality has found numerous uses in theoretical computer science, especially in conjunction with symmetrization as a method for bounding
the $\eps$-approximate degree of various functions (e.g. \cite{paturi, beigelomb, patternmatrix, klivanssherstov, aaronsonshi, NS94, agnostic, colt}).

We prove a number of important special cases of 
this inequality based on linear programming duality. 
Our proofs are constructive in that we
exhibit explicit dual solutions to a linear program bounding the derivative of a constrained polynomial.

The special cases of the Markov-Bernstein inequality that we prove are sufficient for many applications in theoretical computer science.
The dual solutions we exhibit are remarkably clean, and we believe that they shed new light on these classical inequalities.

\subsection{Proving the Markov-Bernstein Inequality at $x=0$}
The following linear program with uncountably many constraints captures the problem of finding 
a polynomial $p(x) = c_n x^n + c_{n-1} x^{n-1} + \dots + c_1 x + c_0$ with real-valued coefficients 
that maximizes $|p'(0)|$ subject to the constraint that $\|p\|_{[-1, 1]}\leq 1$. 
Below
the variables are $c_0, \dots c_n$, and there is a constraint for every $x \in [-1, 1]$. To handle the case where $i = 0$ and $x=0$, we use the convention $0^0 = 1$.

\[ \boxed{\begin{array}{rll} 
    \text{max}      & c_1    \\
    \mbox{such that} & \sum_{i=0}^n c_i x^i \leq 1, \mbox{ }  \forall x \in [-1,1] \\
              	&  - \sum_{i=0}^n c_i x^i \leq 1, \mbox{ }  \forall x \in [-1,1] 
    \end{array}}
\]

One might initially be concerned that our goal is to bound $|p'(x)|$, while the above LP only yields an upper bound on $p'(x)$.
But for any polynomial $p$ satisfying $\|p\|_{[-1, 1]} \leq 1$ whose derivative is negative, $-p$ is a feasible solution
to the above LP achieving value $|p'(x)|$. Thus, the value of the above LP indeed equals $\sup_{p \in B_n} |p'(0)|$, where $B$ denotes the set of all degree $n$ polynomials $p$ satisfying $\|p'\|_{[-1, 1]} \leq 1$.

We will actually upper bound the value of the following LP, which is obtained
from the above by throwing away all but finitely many constraints. Not coincidentally, the constraints 
that we keep are those that are tight for the primal solution corresponding to the Chebyshev polynomials of the 
first kind. Throughout this section, we refer to this LP as {\sc Primal}.

\[ \boxed{\begin{array}{rl} 
    \text{max}      & c_1   \\
    \mbox{such that} &   \sum_{i=0}^n c_i x^i \leq 1, \mbox{ } \forall x=\cos(j\pi/n),  \mbox{ }  j\in \left\{0, 2, \dots, n-1\right\} \\
              	&   - \sum_{i=0}^n c_i x^i \leq 1,\mbox{ }  \forall x = \cos(j\pi/n),  \mbox{ }  j\in \left\{1, 3, \dots, n\right\}
    \end{array}
    }
\]

The dual to {\sc Primal} can be written as 

\[ \boxed{\begin{array}{rl} 
    \text{min}     & \sum_{i=0}^n y_i     \\
    \mbox{such that} &  Ay = e_1\\
    & y_j \geq 0 \mbox{ } \forall j \in \{0, \dots, n\}
    \end{array}
    }
\]

where $A_{ij} = (-1)^j \cos^i(j \pi/n)$ and $e_1=(0, 1, 0, 0, 0, \dots, 0)^T$, again taking $0^0 = 1$. We refer to this linear program as {\sc Dual}.

Our goal is to prove that {\sc Primal} has value at most $n$. For odd $n$, it is well-known that this value is achieved by the coefficients of 
$(-1)^{(n-1)/2} T_n(x)$, the degree $n$ Chebyshev polynomial of the first kind.
Our knowledge of this primal-optimal solution informed our search for a dual-optimal solution, 
but our proof makes no explicit reference to the Chebyshev polynomials, and we do not need to invoke strong LP duality; weak duality suffices. 

Our arguments make use of a number of trigonometric identities that can all be 
established by elementary methods. These identities are presented in \appref{app:markovbernstein}.

\begin{proposition}\label{prop:e1-soln}
Let $n = 2m + 1$ be odd. Define the $(n+1) \times (n+1)$ matrix $A$ by $A_{ij} = (-1)^{j + m} \cos^i(j\pi /n)$ for $0 \le i, j \le n$. Then
\[y = \frac{1}{n}(1/2, \sec^2(\pi/n), \sec^2(2\pi/n), \dots, \sec^2((n-1)\pi/n), 1/2)^T\]
is the unique solution to $Ay = e_1$, where $e_1 = (0, 1, 0, 0, \dots, 0)^T$.
\end{proposition}

Before proving the proposition, we explain its consequences. Note that $y$ is clearly nonnegative, and thus is the unique feasible solution for {\sc Dual}. Therefore it is the dual-optimal solution, and exactly recovers the Markov-Bernstein inequality at $x = 0$:

\begin{corollary} \label{cor:odd-deriv-zero}
Let $p$ be a polynomial of degree $n = 2m + 1$ with $\|p\|_{[-1, 1]} \le 1$. Then $p'(0) \le n$.
\end{corollary}

\begin{proof}
Let $y$ be as in \propref{prop:e1-soln}. This is the unique feasible point for {\sc Dual}. By \lemref{lem:sec-sq-sum} in \appref{app:markovbernstein},
\[\sum_{j=0}^{n-1} \sec^2\left(\frac{j\pi}{n}\right) = n^2,\]
so we immediately see that $\sum_{j=0}^n y_j = n$. By weak LP duality, the
value of {\sc Primal} is at most $n$. 
\end{proof}

While we have recovered the Markov-Bernstein inequality only for odd-degree polynomials at the point $x=0$, a simple ``shift-and-scale'' argument recovers the asymptotic bound for any $x$ bounded away from
the endpoints $\{-1, 1\}$.

\begin{corollary} \label{cor:scale}
Let $p$ be a polynomial of degree $n$ with $\|p\|_{[-1, 1]} \le 1$. Then for any $x_0 \in (-1, 1)$, $|p'(x_0)| \le \frac{n+1}{1-|x_0|}\|p\|_{[-1, 1]}$.
In particular, for any constant $\eps \in (0, 1)$, $\|p'\|_{[-1 + \eps, 1 - \eps]} = O(n)\|p\|_{[-1, 1]}$. 
\end{corollary}

\begin{proof}
Assume without loss of generality that $x_0 \in [0, 1)$ -- an identical argument holds if $x_0 \in (-1, 0]$. By \corref{cor:odd-deriv-zero}, $|q'(0)| \le (n + 1) \|q\|_{[-1, 1]}$ for any polynomial $q$ of degree at most $n$. Define the degree-$n$ polynomial $q(x) = p((1 - x_0)x + x_0)$. Since $(1 - x_0)x + x_0 \in [-1, 1]$ for every $x \in [-1, 1]$, we have $\|q\|_{[-1, 1]} \le \|p\|_{[-1, 1]}$. Moreover, $q'(x) = p'((1 - x_0)x + x_0)(1-x_0)$, so $q'(0) = p'(x_0)(1 - x_0)$. Therefore,
\[|p'(x_0)| = \frac{|q'(0)|}{1 - |x_0|} \le  \frac{n+1}{1 - |x_0|} \|p\|_{[-1, 1]}.   \]
\end{proof}

We remark that the full Markov-Bernstein inequality guarantees that $|p'(x)| \le \frac{n}{\sqrt{1 - x^2}} \|p\|_{[-1, 1]}$, which has quadratically better dependence on the distance from $x$ to $\pm 1$. However, for $x$ bounded away from $\pm 1$ our bound is asymptotically tight and sufficient for many applications in theoretical computer science. Moreover, we can recover the Markov-Bernstein inequality near $\pm 1$ by considering a different linear program (cf. \subsecref{subsec:markovbernstein1}).
\medskip

\begin{proofof}{\propref{prop:e1-soln}}
We write
\begin{equation} \label{eq:prope1soln} (Ay)_i = \frac{(-1)^m}{2n} + \frac{(-1)^{i+m+1}}{2n} + \frac{1}{n}\sum_{j=1}^{n-1} (-1)^{j+m}\cos^{i-2}\left(\frac{j\pi}{n}\right).\end{equation}
Our goal is to show that $(Ay)_i = 1$ for $i = 1$, and $(Ay)_i = 0$ for all other $i$.
The case where $i$ is even is easy. Since $\cos(\pi - \theta) = -\cos\theta$, the terms in the sum naturally pair up. Specifically,
\[(-1)^{j+m}\cos^{i-2}\left(\frac{j\pi}{n}\right) + (-1)^{(n-j)+m}\cos^{i-2}\left(\frac{(n-j)\pi}{n}\right) = 0,\]
so the sum in \eqref{eq:prope1soln} is clearly zero.

Now suppose $i$ is odd and larger than $1$. Then \lemref{lem:odd-cosine-power-sum} in Appendix \ref{app:markovbernstein} implies that $(Ay)_i = 0$. All that remains is the case of $i = 1$. We write the sum explicitly as
\[(Ay)_1 = \frac{(-1)^m}{n} + \frac{1}{n} \sum_{j=1}^{n-1} (-1)^{j+m} \sec \left(\frac{j \pi}{n}\right).\]
By \lemref{lem:sec-sum} in Appendix \ref{app:markovbernstein}, this evaluates to $1$.
\end{proofof}

\subsection{Proving the Markov-Bernstein Inequality at $x=1$} \label{subsec:markovbernstein1}
A similar strategy allows us to bound the derivative of a degree-$n$ polynomial $p$ at the point $x = 1$. We can expand $p(x)$ around $1$ as $p(x) = c_n(x-1)^n + c_{n-1}(x-1)^{n-1} + \dots + c_1(x-1) + c_0$. Then $p'(1) = c_1$. A modest update to {\sc Primal} captures the problem of maximizing $p'(1)$ subject to boundedness constraints at the Chebyshev nodes. 

\[\boxed{ \begin{array}{rl} 
    \max      & c_1   \\
    \mbox{such that} &   \sum_{i=0}^n c_i (x-1)^i \leq 1, \mbox{ } \forall x=\cos(j\pi/n),  \mbox{ }  0 \le j \le n, \ j \text{ even} \\
              	  &-\sum_{i=0}^n c_i (x-1)^i \leq 1,\mbox{ }  \forall x = \cos(j\pi/n),  \mbox{ } 0 \le j \le n, \ j \text{ odd}
    \end{array}}
\]
The dual linear program takes the form
\[ \boxed{ \begin{array}{rl} 
    \min      & \sum_{i=0}^n y_i     \\
    \mbox{such that} &  By = e_1\\
    & y_j \geq 0 \mbox{ } \forall j \in \{0, \dots, n\}
    \end{array}
    }
\]
where $B_{0,0} = 1$ and $B_{ij} = (-1)^j(\cos(j\pi / n) - 1)^i$ otherwise. The determinant of $B$ is, up to sign, a Vandermonde determinant, and in particular is nonzero. Thus, $By = e_1$ has a unique solution. Again, we can write down this solution explicitly.

\begin{proposition} \label{prop:e1-soln-at-1}
Let $n$ be a natural number, and define the $(n + 1) \times (n + 1)$ matrix $B$ as above. Then 
\[y = \left(\frac{2n^2 + 1}{6}, \csc^2\left(\frac{\pi}{2n}\right), \csc^2\left(\frac{2\pi}{2n}\right), \dots, \csc^2\left(\frac{(n-1)\pi}{2n}\right), \frac{1}{2}\right)\]
is the unique solution to $By = e_1$.
\end{proposition}

\begin{proof}
We just need to show that $By = e_1$. First, if $i \ne 0$ then
\[(By)_i = \frac{(-1)^n(-2)^i}{2} + \sum_{j=1}^{n-1} (-1)^j \csc^2 \left(\frac{j\pi}{2n}\right) \left(\cos\left(\frac{j\pi}{n}\right) - 1\right)^i.\]
Using the half-angle identity $\sin^2 (\theta / 2) = (1 - \cos\theta)/2$, this becomes
\[(By)_i = (-1)^{i+n}2^{i-1} + (-1)^i2^i\sum_{j=1}^{n-1} (-1)^{j} \sin^{2i - 2} \left(\frac{j\pi}{2n}\right).\]
If $i = 1$, then the sine terms are identically $1$ so $(By)_1$ evaluates to $1$ (note that the calculation is slightly different depending on whether $n$ is even or odd). If $i > 1$, then by \lemref{lem:even-sine-power-sum} in Appendix \ref{app:markovbernstein}, the sum of sine terms evaluates to $\frac{1}{2}(-1)^n - (-1)^n = -\frac{1}{2} (-1)^n$. Therefore, $(By)_i = 0$ for all $i > 1$.

Finally, we need to show that $(By)_0 = 0$. We expand
\[(By)_0 = \frac{2n^2 + 1}{6} + \frac{1}{2}(-1)^n + \sum_{j=1}^{n-1}(-1)^j \csc^2 \left(\frac{j\pi}{2n}\right).\]
By \lemref{lem:alt-csc-sum}, this evaluates to 0.
\end{proof}

\begin{corollary} \label{cor:markovatone}
If $p$ is a polynomial of degree $n$ with $\|p\|_{[-1, 1]} \le 1$, then $p'(1) \le n^2$.
\end{corollary}

\begin{proof}
Let $y$ be as in \propref{prop:e1-soln-at-1}. Notice that $y_j \geq 0$ for all $j \in \{0, \dots, n\}$. Combined with Proposition \ref{prop:e1-soln-at-1}, it is clear that $y$ is dual-feasible.
By \lemref{lem:csc-sum},
\begin{align*}
\sum_{j=0}^n y_j &= \frac{2n^2 + 1}{6} + \frac{1}{2} + \sum_{j=1}^{n-1} \csc^2 \left(\frac{j\pi}{2n}\right) \\
&= \frac{n^2}{3} + \frac{2}{3} + \frac{4n^2 - 4}{6} = n^2.
\end{align*}
\end{proof}

By combining Corollary \ref{cor:markovatone} with a shifting and scaling argument similar to the one used to prove \corref{cor:scale},
we recover an asymptotic statement of Markov's inequality for the first derivative of a constrained polynomial.

\begin{corollary} \label{cor:scale2}
If $p$ is a polynomial of degree $n$, then for all  $x_0 \in [-1,1]$ with $x_0 \ne 0$, $|p'(x_0)| \le  \frac{n^2}{|x_0|} \|p\|_{[-1,1]}$. Thus, for any constant $\eps \in (0, 1)$, $\|p'\|_{[-1, -\eps] \cup [\eps, 1]} = O(n^2) \|p\|_{[-1, 1]}$.
\end{corollary}

\begin{proof}
The argument is the same as in the proof of \corref{cor:scale}, except we instead use the auxiliary polynomial $q(x) = p(|x_0| x)$.
\end{proof}

Combining this with \corref{cor:scale}, we recover an asymptotically tight version of Markov's inequality for the whole interval $[-1, 1]$.

\begin{corollary} \label{cor:finalmarkov}
If $p$ is a polynomial of degree $n$, then for all  $x \in [-1,1]\texttt{•},$ $|p'(x)| \le  O(n^2) \|p\|_{[-1,1]}$.
\end{corollary}

\subsection{Markov's inequality for higher derivatives}
In 1892, V. Markov proved the following generalization of the Markov-Bernstein inequality to higher derivatives. Let $p$ be a real polynomial of degree at most $n$, and let $T_n$ be the $n$th Chebyshev polynomial of the first kind. Then
\[|p^{(k)}(x)| \le T^{(k)}_n(1) \|p\|_{[-1, 1]}\]
for every $x \in [-1, 1]$. We use complementary slackness to prove an important special case of this inequality, namely that $p^{(k)}(1) \le  T^{(k)}_n(1) \|p\|_{[-1, 1]}$. 

While A. A. Markov's inequality for the first derivative has a short proof (see \cite{bowlinggreen} for a proof using tools from approximation theory), the generalization to higher derivatives is considered a deep theorem \cite{shadrin}. The shortest known proof of this theorem proceeds in two steps \cite[Section 3.1]{shadrin}.
In the first step, it is shown that among all points $x \in [-1, 1]$, the quantity $\sup_{p \in B} |p^{(k)}(x)|$ is maximized at $x=1$, where again $B$ is the set of degree $n$ polynomials $p$ with real coefficients such that $\|p\|_{[-1, 1]} \le 1$. In the second step, it is shown that 
$p^{(k)}(1) \le  T^{(k)}_n(1) \|p\|_{[-1, 1]}$. It is this second step that we prove here using complementary slackness.

The following lemma, found in \cite{heineman}, relates the determinant of a Vandermonde matrix having the degrees of the monomials in its last $(n-k)$ rows incremented by $1$ to the determinant of an ordinary Vandermonde matrix. For each integer $0 \le k \le n$, we define the elementary symmetric polynomial
\[e_k(x_1, \dots, x_n) = \sum_{1 \le j_1 < j_2 < \dots < j_k \le n} x_{j_1} x_{j_2} \dots x_{j_k}.\]

\begin{lemma} \label{lemma:20}
Let $0 \le k \le n$. Then
\[\left|
\begin{array}{cccc}
1 & 1 & 1 & 1 \\
x_1 & x_2 & \dots & x_n \\
\vdots & \vdots & & \vdots \\
x_1^{k-1} & x_2^{k-1} & \dots & x_n^{k-1} \\
x_1^{k+1} & x_2^{k+1} & \dots & x_n^{k+1} \\
\vdots & \vdots & & \vdots \\
x_1^{n} & x_2^{n} & \dots & x_n^{n} \\
\end{array}
\right| = e_{n-k}(x_1, x_2, \dots, x_n)
\left|
\begin{array}{cccc}
1 & 1 & 1 & 1 \\
x_1 & x_2 & \dots & x_n \\
\vdots & \vdots & & \vdots \\
x_1^{k} & x_2^{k} & \dots & x_n^{k} \\
\vdots & \vdots & & \vdots \\
x_1^{n-1} & x_2^{n-1} & \dots & x_n^{n-1} \\
\end{array}
\right|\]
\end{lemma}

\begin{proposition}
Let $p$ be a polynomial of degree $n$ with $|p(x)| \le 1$ for $x \in [-1, 1]$. Then 
\[p^{(k)}(1) \le T^{(k)}_n(1)\]
where $T_n(x)$ is the $n$-th Chebyshev polynomial of the first kind.
\end{proposition}

\begin{proof}
This is obvious if $k = 0$, since $T_n(1) = 1$, so we assume $k > 0$. Recall the expansion $p(x) = c_n(x-1)^n + c_{n-1}(x-1)^{n-1} + \dots + c_1(x-1) + c_0$. Then the $k$-th derivative of $p$ at $1$ is simply $k!c_k$. We consider the linear program

\[ \boxed{\begin{array}{rl} 
    \max      & k!c_k   \\
    \mbox{such that} &   \sum_{i=0}^n c_i (x-1)^i \leq 1, \mbox{ } \forall x=\cos(j\pi/n),  \mbox{ }  0 \le j \le n, \ j \text{ even} \\
              	  &-\sum_{i=0}^n c_i (x-1)^i \leq 1,\mbox{ }  \forall x = \cos(j\pi/n),  \mbox{ } 0 \le j \le n, \ j \text{ odd}
    \end{array}}
\]
and its dual
\[ \boxed{\begin{array}{rl} 
    \min      & \sum_{i=0}^n y_i     \\
    \mbox{such that} &  By = k!e_k\\
    & y_j \geq 0 \mbox{ } \forall j \in \{0, \dots, n\}
    \end{array}}
\]
where $B_{0,0} = 1$ and $B_{ij} = (-1)^j(\cos(j\pi / n) - 1)^i$ otherwise. Notice that all primal constraints are tight for the primal solution corresponding to $T_n$, the degree $n$ Chebyshev polynomial of the first kind.

The determinant of $B$ is, up to sign, a Vandermonde determinant, and in particular is nonzero. Thus, $By = k!e_k$ has a unique solution. If we can show that this solution has positive entries, complementary slackness (cf. \cite[pg. 95]{schrijver})
implies that $T_n$ is a primal optimal solution, and the result will follow.

We now use Cramer's rule to investigate the solution to $By = k!e_k$. Recall that Cramer's rule tells us that entry $y_j$ is given by $\det B_j / \det B$ where the matrix $B_j$ is obtained from $B$ by replacing its $j$th column with $k!e_k$. Using the formula for the Vandermonde determinant, $\det B$ is given by
\[(-1)^{\lfloor (n+1)/2 \rfloor}\prod_{0 \le j < j' \le n} \left(\cos\left(\frac{j'\pi}{n}\right) - \cos\left(\frac{j\pi}{n}\right)\right).\]
Since $\cos(x)$ is  a decreasing function on the interval $[0, \pi]$, all the terms in the product are negative. Therefore, the sign of $\det B$ is $(-1)^{\lfloor (n+1)/2 \rfloor + {n + 1\choose 2}}$.

For convenience, let $\alpha_j = \cos(j\pi / n) - 1$. Consider the numerator of Cramer's rule for entry $y_j$. This is the determinant of the matrix $B_j$,
\[\begin{pmatrix}
1 & -1 & \dots & (-1)^{j-1} & 0 & (-1)^{j+1} & \dots & (-1)^n \\
0 & -\alpha_1  & \dots & (-1)^{j-1}\alpha_{j-1} & 0 & (-1)^{j+1}\alpha_{j+1} & \dots & (-1)^{n}(-2) \\
\vdots & \vdots & & \vdots & \vdots & \vdots & & \vdots \\
0 & -\alpha_1^k  & \dots & (-1)^{j-1}\alpha_{j-1}^k & k! & (-1)^{j+1}\alpha_{j+1}^k &\dots & (-1)^{n}(-2)^k \\
\vdots & \vdots & & \vdots & \vdots & \vdots & & \vdots \\
0 & -\alpha_1^n  & \dots & (-1)^{j-1}\alpha_{j-1}^n & 0 & (-1)^{j+1}\alpha_{j+1}^n & \dots & (-1)^{n}(-2)^n \\
\end{pmatrix}.\]
Taking the cofactor expansion along the replaced column, and factoring out $-1$ from each of the appropriate columns gives
\[k!(-1)^{\lfloor (n+1)/2 \rfloor+j} \cdot (-1)^{j+k}\left|
\begin{array}{cccccc}
1 & \dots & 1 & 1 & \dots & 1 \\
0 & \dots & \alpha_{j-1} & \alpha_{j+1} & \dots & -2 \\
\vdots & & \vdots & \vdots & \dots & \vdots \\
0 & \dots & \alpha_{j-1}^{k-1} & \alpha_{j+1}^{k-1} & \dots & (-2)^{k-1} \\
0 & \dots & \alpha_{j-1}^{k+1} & \alpha_{j+1}^{k+1} & \dots & (-2)^{k+1} \\
\vdots & & \vdots & \vdots & \dots & \vdots \\
0 & \dots & \alpha_{j-1}^{n} & \alpha_{j+1}^{n} & \dots & (-2)^{n} \\
\end{array}
\right|.\]
The matrix satisfies the conditions of Lemma \ref{lemma:20}, so we can write this as
\[k!(-1)^{\lfloor (n+1)/2 \rfloor+k}e_{n-k}(\alpha_0, \dots, \alpha_{j-1}, \alpha_{j+1}, \dots, \alpha_n) \prod_{\substack{0 \le i < i' \le n \\ i, i' \ne j}} (\alpha_{i'} - \alpha_i).\]
There are ${n \choose 2}$ strictly negative terms in the product, and as long as $k > 0$, $e_{n-k}$ has sign $(-1)^{n-k}$. So the sign of the whole product is $(-1)^{\lfloor (n+1)/2 \rfloor + n + {n \choose 2}}$. Dividing by the sign of $\det B$, we get $(-1)^{n + {n \choose 2} - {n+1 \choose 2}} = 1$.
\end{proof}

\section{Conclusion}
The approximate degree is a fundamental measure of the complexity of a Boolean function, with pervasive applications throughout theoretical computer science.
We have sought to advance our understanding of this complexity measure by resolving the approximate degree of the AND-OR tree, and reproving known lower bounds through
the construction of explicit dual witnesses. Nonetheless, few general results on approximate degree are known, and many interesting open questions remain. 

\begin{itemize}
\item Our understanding of the approximate degree of fundamental classes of functions remains incomplete. For example, the approximate degree of $\text{AC}^0$ remains open \cite{beame, aaronsonshi}: the best known lower bound is $\tilde{\Omega}(n^{2/3})$ \cite{aaronsonshi}, while no $o(n)$ upper bound is known. It is also open to determine the least approximate degree of any ``approximate majority'' function
(see \cite[Page 11]{openproblems}).\footnote{This open problem is due to Srikanth Srinivasan.} 

\item While polynomial relationships are known between approximate degree and other complexity measures such as decision-tree depth, exact degree (i.e. $\deg_0$), and block sensitivity, it is still open to determine the largest possible gaps between these quantities. For instance, the exact degree of the $\text{OR}_n$ function is $n$, exhibiting a quadratic gap between $\deg_0$ and $\widetilde{\deg}$, which is the largest known. Is this separation the best possible?

\item Finally, the proof of our lower bound on the approximate degree of the \text{AND-OR} tree relied crucially on the fact that a dual polynomial for OR has one-sided error. This same observation was used by Gavinsky and Sherstov \cite{gavinskysherstov} to separate the multiparty communication versions of NP and co-NP, and very recently by the current authors \cite{bunthaler} to derive new discrepancy and threshold weight bounds for $\text{AC}^0$. What other functions have dual polynomials with one-sided error, and are there further applications for these objects?
\end{itemize}

Resolving these open questions may require moving beyond traditional symmetrization-based
arguments, which transform a polynomial $p$ on $n$ variables into a polynomial $q$ on $m < n$ variables in such a way that $\widetilde{\deg}(q) \leq \widetilde{\deg}(p)$, before obtaining 
a lower bound on $\widetilde{\deg}(q)$. Symmetrization necessarily ``throws away'' information about $p$; in contrast, the method of constructing dual polynomials appears to be a very powerful and complete way of reasoning about approximate degree. 
Can progress be made on these open problems by directly constructing good dual polynomials?

\paragraph{Acknowledgements.} We are grateful to Ryan O'Donnell and Li-Yang Tan for posing the problem of proving Markov-type inequalities via the construction of a dual witness,
and to Karthekeyan Chandrasekaran, Troy Lee, Robert \v{S}palek, Jon Ullman, Andrew Wan, and the anonymous ICALP reviewers for valuable feedback on early versions of this manuscript. 

\appendix

\section{Final Details of \thmref{thm:andor}}
\label{app:andor}
\subsection{Proof of Proposition \ref{finalprop}}
\label{app:andor1}
\begin{proof}[]
Let $r = \lfloor 1 / \alpha\rfloor$. Then
\begin{equation} \label{eqsubsetbound}
\mathbf{P}_y[F(z) \ne F(z_1y_1, \dots, z_My_M)] \leq \mathbf{P}_y[F(z) \ne F(z_1w_1, \dots, z_Mw_M) \text{ for some } w \preceq y]
\end{equation}
where $w \preceq y$ if $\{i : w_i = -1\} \subseteq \{i : y_i = -1\}$. By monotonicity, it suffices to bound the right hand side under the assumption that each bit of $y$ takes the value $-1$ independently with probability exactly $1/r$.

Consider a matrix $Y \in \{-1, 1\}^{r \times M}$ where each column is chosen independently at random from the $r$ vectors having a $-1$ in one slot and a $+1$ in all the others. Let $y^1, y^2, \dots, y^r$ denote the rows of $Y$. While these rows are not independent, each is individually a random string whose $i$th bit independently takes the value $-1$ with probability $1/r$. Thus the right-hand side of Expression (\ref{eqsubsetbound}) equals
\begin{align*}
\frac{1}{r} \sum_{j=1}^r \mathbf{P}_Y[F(z) \ne &F(z_1w_1, \dots, z_Mw_M) \text{ for some } w \preceq y^j] \\
&=\frac{1}{r}\mathbf{E}_Y \left[\#\{j : F(z) \ne F(z_1w_1, \dots, z_Mw_M) \text{ for some } w \preceq y^j\}\right]
\end{align*}
The latter count has at most $\operatorname{bs}_z(F)$ nonzero terms because $y^1, \dots, y^j$ are the characteristic vectors of disjoint sets. The asserted inequality follows because $1/r = 1/\lfloor 1/\alpha \rfloor \le 2\alpha$.
\end{proof}

\subsection{Proof of Equation \ref{eq:show3}}
\label{app:andor2}
We prove that the polynomial $\zeta$ defined in \eqref{eq:zeta} satisfies \eqref{eq:show3}, reproduced here for convenience.

$$ \sum_{(x_1, \dots, x_{M}) \in \left(\{-1, 1\}^{N}\right)^M}  \zeta(x_1, \dots, x_{M})\chi_S(x_1, \dots, x_M) =0   \text{ for each } |S| \le d\cdot d'.\quad \quad (\ref{eq:show3})$$

To prove \eqref{eq:show3}, notice that since $\Psi$ is orthogonal on $\{-1, 1\}^M$ to all polynomials of degree at most $d$, we have the Fourier representation
$$\Psi(z) =	\sum_{\substack{T \subseteq \{1, \dots, M\} \\ |T| > d}}	\hat{\Psi}(T) \chi_T(z)$$
for some reals $\hat{\Psi}(T)$. We can thus write
\[\zeta(x_1, \dots, x_M) = 2^M \sum_{|T| > d} \hat{\Psi}(T) \prod_{i \in T} \psi(x_i)\prod_{i \notin T} |\psi(x_i)|.\]
Given a subset $S \subseteq \{1, \dots, M\} \times \{1, \dots, N\}$ with $|S| \le d \cdot d'$, partition $S = (\{1\} \times S_1) \cup \dots \cup (\{M\} \times S_M)$ where each $S_i \subseteq \{1, \dots, N\}$. Then
\begin{align*}
\sum_{(x_1, \dots, x_{M}) \in \left(\{-1, 1\}^{N}\right)^M}& \zeta(x_1, \dots, x_{M})\chi_S(x_1, \dots, x_M) \\
&= 2^M \sum_{|T| > d} \hat{\Psi}(T) \prod_{i \in T} \underbrace{\left( \sum_{x_i \in \{-1, 1\}^N} \psi(x_i) \chi_{S_i}(x_i)\right)} \prod_{i \notin T} \left(\sum_{x_i \in \{-1, 1\}^N} |\psi(x_i)| \chi_{S_i}(x_i)\right).
\end{align*}
Since $|S| \le d \cdot d'$, by the pigeonhole principle, $|S_i| \le d'$ for at least $M - d$ indices $i \in \{1, \dots, M\}$. Thus for each set $T$, at least one of the underbraced factors is zero, as $\chi_{S_i}$ is orthogonal to $\psi$ whenever $|S_i| \le d'$.

\section{Dual Polynomials for Symmetric Functions} \label{app:sym-dual}

\begin{proofof}{\lemref{lemma:min-block}}
Fix an $\ell$ such that $\ell \in [\lfloor ct \rfloor]$ and let $i(k) = tk^2 + 4\ell$. It is enough to show that $\prod_{i' \in S, i' \ne i} |i - i'|$ is minimized at $k = 1$. We can expand this product as
\[\prod_{\substack{i' \in S \\ i' \ne i}} |i - i'| = \prod_{m = 0}^{\lfloor ct \rfloor} \left((tk^2 + 4\ell - (t - 4m)) \prod_{\substack{j = 1 \\ j \ne k}} ^ {\lfloor \sqrt{(n-t+1)/t} \rfloor} |tk^2 + 4\ell - (tj^2 + 4m)|\right) \times \prod_{\substack{m = 0 \\ m \ne \ell}} ^ {\lfloor ct \rfloor} |4\ell - 4m|.\]
Cancelling the factor independent of $k$ and considering each index $m$ separately, we just need to show that for any fixed $0 \le \ell, m \le ct$, the product
\[(tk^2 + 4\ell - (t - 4m))\prod_{j \ne k} |tk^2 + 4\ell - (tj^2 + 4m)|\]
as a function of $k \ge 1$ is minimized at $k = 1$. Divide each factor by $t$ to obtain
\begin{equation}\label{perturbed_prod}
\left(k^2 - 1 + \frac{4(\ell + m)}{t}\right) \prod_{j \ne k} |k^2 - j^2| \left(1 + \frac{4(\ell - m)}{t(k^2 - j^2)}\right).
\end{equation}

We first obtain a lower bound for this expression when $k \ge 2$. Consider the following two facts.

\begin{fact} \label{fact:min-prod}
Let $k \le m$ be nonnegative integers. Then
\begin{equation} \label{eqlemmaminprod}  \prod_{\substack{j \in [m] \\ j \ne k}} |k^2 - j^2| \ge \prod_{\substack{j \in [m] \\ j \ne 1}} |1 - j^2|.\end{equation}
In other words, this product of differences of squares is minimized at $k = 1$.
\end{fact}

\begin{proof}[Proof of Fact \ref{fact:min-prod}]
This is clear if $k = 0$, so suppose $k \ge 2$. Then the left-hand side of Expression (\ref{eqlemmaminprod}) can be written as
\[\prod_{\substack{j \in [m] \\ j \ne k}} |k^2 - j^2| = \prod_{\substack{j \in [m] \\ j \ne k}} (k + j)|k - j| = \frac{(m + k)!}{2k(k-1)!} \cdot k!(m - k)! = \frac{1}{2} (m + k)!(m-k)!.\]
Taking the ratio of the left-hand side of Expression (\ref{eqlemmaminprod}) to the right gives us
\[\frac{(m+k)!(m-k)!}{(m+1)!(m-1)!} = \frac{(m+k)(m+k-1)\dots(m+2)}{(m-1)(m-2)\dots(m-k+1)}\]
which is a product of numbers that are all at least $1$.
\end{proof}

\begin{fact} \label{fact:thisveryfact}
Let $k$ be a nonnegative integer. Then
\[\sum_{\substack{j \in \mathbb{Z} \\ j \ne k}} \frac{1}{|j^2 - k^2|} \le \frac{\pi^2}{3}.\]
\end{fact}

\begin{proof}[Proof of Fact \ref{fact:thisveryfact}]
First suppose $j > k$. Then
\[j^2 - k^2 = (j-k)^2 + 2jk - 2k^2  > (j - k)^2.\]
Thus
\[\sum_{j > k} \frac{1}{j^2 - k^2} < \sum_{j > k} \frac{1}{(j-k)^2} = \frac{\pi^2}{6}.\]
A similar argument holds for $j < k$.
\end{proof}

Combining the two facts, Expression (\ref{perturbed_prod}) is at least

\begin{align*}
(k^2 - 1) \prod_{j \ne k} \left(1 - \frac{4c}{|k^2 - j^2|}\right) \prod_{j \ne 1} |1 - j^2| &\ge \left(1 - \sum_{j \ne k} \frac{4c}{|k^2 - j^2|}\right)\prod_{j \ne 1} |1 - j^2| \\
&\ge \left(1 - \frac{4c\pi^2}{3}\right)\prod_{j \ne 1} |1 - j^2|
\end{align*}
whenever $k \ge 2$. On the other hand, setting $k = 1$ in Expression (\ref{perturbed_prod}) gives us at most
\begin{align*}
8c \prod_{j\ne 1} |1 - j^2| \left(1 + \frac{4c}{|1 - j^2|}\right) &\le 8c\exp\left(\sum_{j \ne 1} \frac{4c}{|1 - j^2|}\right)\prod_{j\ne 1} |1 - j^2| \\
&\le 8c\exp\left(\frac{4c\pi^2}{3}\right)\prod_{j\ne 1} |1 - j^2| \\
\end{align*}
which is easily verified to be smaller than our lower bound for the $k \ge 2$ case if $c \le 1/32$.
\end{proofof}

\begin{proofof}{\lemref{lemma:term-bound}}
Write
\begin{align*}
|P(r)| &= \frac{\pi_S(i^*)}{n!} \frac{\prod_{j \in [n] \setminus \{r\}} |r - j|}{|r - (i^*-1)| |r - (i^* + 1)|\prod_{j \in S \setminus \{r\}} |r - j|} \\
&\le \frac{1}{n!} \frac{r!(n-r)!}{|r - (i^*-1)| |r - (i^* + 1)|} & \text{by definition of $i^*$}\\
&\le \frac{1}{{n \choose r}} \frac{1}{t^2(k^2 - (1 +4c + 1/t))^2}.
\end{align*}
The bound follows since $c \le 1/32$ and $t \ge 2$. The calculation for $v$ is similar.
\end{proofof}

\begin{proofof}{\lemref{lemma:tail-bound}}
Using the bounds from the previous lemma, as well as the facts that $c \le 1/32$ and $t \ge 2$, the left hand side is at most
\begin{align*}
\sum_{\ell \ne 0} \frac{1}{16\ell^2 - 1} + \sum_{k \ge 2} \sum_{\ell = 0}^{\lfloor ct \rfloor} \frac{1}{t^2(k^2-2)^2} &\le 2 \sum_{\ell = 1}^\infty \frac{1}{15\ell^2} + \frac{ct+1}{t^2} \sum_{k=2}^\infty \frac{1}{(k^2 - 2)^2} \\ 
& \le \frac{\pi^2}{45} +  \frac{17}{64} \sum_{k=2}^\infty \frac{4}{k^4}\\
&= \frac{\pi^2}{45} + \frac{17}{16} \left(\frac{\pi^4}{90} - 1\right) \\
&\le \frac{2}{5}.
\end{align*}
\end{proofof}

\begin{proofof}{\propref{prop:sym-lb}}
By symmetry, we can assume that $t \le n /2$. Moreover, we may assume that $t$ is the largest such integer with $F(t-1) \ne F(t)$. We have already handled a few special cases: The case of $t = 1$ corresponds to \v{S}palek's construction for the OR function \cite{spalek}, and the case of $t = \Omega(n)$ follows from Proposition \ref{prop:maj-lb}. We can therefore assume that $2 \le t \le n/4$. We now consider three separate cases based on which of the terms ${n \choose i^*-1} |P(i^*-1)|, 1, {n \choose i^*+1} |P(i^*+1)|$ is the smallest. In all of these cases, we will show that we can construct a polynomial $Q$ such that $(Q \cdot F) / \|Q\|_1 \ge 1/14$.

\paragraph{Case 1:} ${n \choose i^*-1} |P(i^*-1)|, 1 \ge {n \choose i^*+1} |P(i^*+1)|$.

Recall that by translating $S$ by at most $4ct$ (thereby keeping it a subset of $[n]$), we can assume that $i^* = t$. Let $Q(i) = (-1)^iP(i)$. Then the multilinear polynomial associated to $Q$ has pure high degree $|T| = \Omega(\sqrt{t(n-t+1)})$. The $\ell_1$ norm of $Q$ is
\begin{align*}
\|Q\|_1 &= \sum_{i \in T} {n \choose i} |Q(i)| \\
&\le {n \choose t-1}|Q(t-1)| + 1 + {n \choose t+1}|Q(t+1)| + \frac{2}{5} & \text{by \lemref{lemma:tail-bound}}\\
&\le 2{n \choose t-1}|Q(t-1)| + \frac{7}{5}.
\end{align*}
Choose the sign bit $s$ in the definition of $P$ so that $Q(t) = F(t)$. Since $P(t-1)$ has the same sign as $P(t)$, it holds that $\sgn(Q(t-1)) = \sgn(F(t-1))$. Therefore,
\begin{align*}
Q \cdot F &= {n \choose t-1}|Q(t-1)| + 1 + \sum_{i \in T \setminus \{t-1, t\}}  {n \choose i} F(i) Q(i) \\
&\ge {n \choose t-1}|Q(t-1)| + 1 - {n \choose t+1}|Q(t+1)| - \frac{2}{5} \\
&\ge \frac{1}{2} {n \choose t-1}|Q(t-1)| + \frac{1}{2} - \frac{2}{5} \\
&= \frac{1}{2} {n \choose t-1}|Q(t-1)| + \frac{1}{10}.
\end{align*}
Using the fact that $(A + B) / (C + D) \ge \min(A/C, B/D)$ for positive $A, B, C, D$,
\[\frac{Q \cdot F}{\|Q\|_1} \ge \frac{1}{14}.\]

\paragraph{Case 2:} $1, {n \choose i^*+1} |P(i^*+1)| \ge {n \choose i^*-1} |P(i^*-1)|$.

This time, translate $S$ so that $i^* = t-1$. We remark that under this translation we still have $T \subseteq [n]$, since we assumed $t \ge 2$. The remainder of the analysis is identical to Case 1, interchanging the roles of $t-1$ and $t+1$.

\paragraph{Case 3:} ${n \choose i^*-1} |P(i^*-1)|, {n \choose i^*+1} |P(i^*+1)| \ge 1$.

Translate $S$ so that $i^* = t + 1$, and choose $s$ so that $Q(t) = (-1)^{i^*-1}P(i^*-1) = F(t)$. Observe that $F(t+2) = F(t)$, since we chose $t \leq n/4$ to be the largest such integer with $F(t-1) \neq F(t)$. Then $Q(t+2) = (-1)^{i^*+1}P(i^*+1)$ has the same sign as $F(t+2)$. The $\ell_1$ norm calculation follows as in Case 1 to give
\[\|Q\|_1 \le {n \choose t}|Q(t)| + {n \choose t+2}|Q(t+2)| + \frac{7}{5} \le \frac{17}{10}{n \choose t}|Q(t)| + \frac{17}{10}{n \choose t+2}|Q(t+2)|.\]
The correlation with $F$ is
\begin{align*}
Q \cdot F &= {n \choose t}|Q(t)| + {n \choose t+2}|Q(t+2)| + \sum_{i \in T \setminus \{t, t+2\}}  {n \choose i} F(i) Q(i) \\
&\ge {n \choose t}|Q(t)| + {n \choose t+2}|Q(t+2)| - \frac{7}{5} \\
&\ge \frac{3}{10}{n \choose t}|Q(t)| + \frac{3}{10}{n \choose t+2}|Q(t+2)|,
\end{align*}
so $(Q \cdot F) / \|Q\|_1 \ge 3/17$.
\end{proofof}

\section{Index of Trigonometric Identities}
\label{app:markovbernstein}

\begin{lemma} \label{lem:sec-sq-sum}
Let $n$ be odd. Then
\[\sum_{k=0}^{n-1} \sec^2 \left(\frac{k\pi}{n}\right) = n^2.\]
\end{lemma}

\begin{proof}
We start with the identity \cite[No. 445]{jolley}
\[\sum_{k=0}^{n-1} \tan^2 \left(\theta + \frac{k\pi}{n}\right) = n^2 \cot \left(\frac{n\pi}{2} + n\theta\right) + n(n-1).\]
Letting $\theta = 0$, this evaluates to $n(n-1)$ as long as $n$ is odd. Substituting $\tan^2(k\pi / n) = \sec^2(k\pi / n) - 1$ into the left-hand side gives the identity.
\end{proof}

\begin{lemma}[{\cite[No. 429]{jolley}}] \label{lem:alt-cosine}
\[\sum_{j=0}^n (-1)^j \cos\left(j\theta\right) = \frac{1}{2} + (-1)^n\frac{\cos((n+1/2)\theta)}{2\cos(\theta/2)}.\]
\end{lemma}

\begin{lemma}\label{lem:odd-cosine-power-sum}
Let $i < n$ be odd natural numbers. Then
\[\sum_{j=0}^n (-1)^j \cos^i \left(\frac{j\pi}{n}\right) = 1.\]
\end{lemma}

\begin{proof}
For odd $i$, consider the well-known power reduction formula
\[\cos^i\theta = 2^{1-i}\sum_{k=0}^{(i-1)/2} {i\choose k} \cos((i-2k)\theta).\]
Applying the previous lemma with $\theta = (i-2k)\pi / n$,
\begin{align*}
\sum_{j=0}^n (-1)^j \cos^i \left(\frac{j\pi}{n}\right) &= 2^{1-i} \sum_{k = 0}^{(i-1)/2} {i \choose k} \sum_{j=0}^n (-1)^j \cos \left(\frac{(i-2k)j\pi}{n}\right) \\
&= 2^{1-i} \sum_{k = 0}^{(i-1)/2} {i \choose k} \left(\frac{1}{2} + (-1)^n \frac{\cos((i-2k)\pi / 2n + (i-2k)\pi)}{2\cos((i-2k)\pi/2n)}\right). \\
&= 2^{1-i} \sum_{k = 0}^{(i-1)/2} {i \choose k} = 1.
\end{align*}
\end{proof}

\begin{lemma} \label{lem:sec-sum}
Let $n = 2m+1$ be odd. Then
\[\sum_{k=0}^n (-1)^k \sec \left(\frac{k\pi}{n}\right) = (-1)^mn + 1.\]
\end{lemma}

\begin{proof}
This follows from the identity \cite{wolfram}
\[\sum_{k=0}^m \sec \left(\frac{2k\pi}{2m + 1}\right) = \frac{1}{2}(-1)^m(2m+1) + \frac{1}{2},\]
and the observation that $\sec(2k\pi/n) = -\sec((n-2k)\pi / n)$.
\end{proof}

\begin{lemma}\label{lem:even-sine-power-sum}
Let $2 \le 2i < n$. Then
\[\sum_{j=0}^{n} (-1)^j \sin^{2i} \left(\frac{j\pi}{2n}\right) = \frac{1}{2}(-1)^n.\]
\end{lemma}

\begin{proof}
Consider the power reduction formula
\[\sin^{2i}(\theta) = 2^{-2i}{2i \choose i} + 2^{1-2i} \sum_{k=0}^{i-1} (-1)^{i-k} {2i \choose k} \cos ((2i - 2k)\theta). \]
Let $\theta = j\pi/2n$. Then
\[\sum_{j=0}^{n} (-1)^j \sin^{2i} \left(\frac{j\pi}{2n}\right) = \frac{1}{2}(1 + (-1)^n)2^{-2i}{2i \choose i} + 2^{1-2i} \sum_{k=0}^{i-1} (-1)^{i-k} {2i \choose k} \sum_{j=0}^n (-1)^j \cos \left(\frac{(i - k)j\pi}{n}\right).\]
By \lemref{lem:alt-cosine}, the sum on the right simplifies to
\begin{align*}
2^{1-2i} \sum_{k=0}^{i-1} &(-1)^{i-k} {2i \choose k} \left(\frac{1}{2} + (-1)^n \frac{\cos((i-k)\pi/2n + (i-k)\pi)}{2 \cos((i-k)\pi/2n)}\right) \\
&= 2^{1-2i} \sum_{k=0}^{i-1} (-1)^{i-k} {2i \choose k} \left(\frac{1}{2} + \frac{1}{2}(-1)^{n + i-k}\right) \\
&= 2^{-2i} \sum_{k=0}^{i-1} (-1)^{i-k} {2i \choose k} + (-1)^n2^{-2i} \sum_{k=0}^{i-1} {2i \choose k}\\
\end{align*}
Using the identity
\[\sum_{k=0}^{2i} (-1)^{k} {2i \choose k} = 0\]
and the symmetry of the binomial coefficients, the first sum evaluates to $-\frac{1}{2} {2i \choose i}$. Therefore,
\begin{align*}
\sum_{j=0}^{n} (-1)^j \sin^{2i} \left(\frac{j\pi}{2n}\right) &= (1 + (-1)^n)2^{-1-2i}{2i \choose i} - 2^{-1-2i}{2i \choose i} + (-1)^n2^{-2i} \left(2^{2i-1} - \frac{1}{2} {2i \choose i}\right)\\
&= \frac{1}{2}(-1)^n.
\end{align*}
\end{proof}

\begin{lemma} [{\cite[No. 440]{jolley}}] \label{lem:csc-sum}
\[\sum_{j=1}^{n-1} \csc^2 \left(\frac{j\pi}{2n}\right) = \frac{4n^2 - 4}{6}.\]
\end{lemma}

\begin{lemma} \label{lem:alt-csc-sum}
\[\sum_{j=1}^{n-1}(-1)^j \csc^2 \left(\frac{j\pi}{2n}\right) = -\frac{n^2}{3}  - \frac{1}{6} - \frac{1}{2}(-1)^n.\]
\end{lemma}

\begin{proof}
Consider the identity \cite[Nos. 441, 442]{jolley}
\[\sum_{\substack{j=1 \\ j \text{ odd}}}^{n-1} \csc^2 \left(\frac{j\pi}{2n}\right) = \frac{n^2}{2}  + \frac{1}{4} ((-1)^n - 1).\]
Let $\theta = \pi/2n$ and subtract twice the second identity from the identity in \lemref{lem:csc-sum}. Then we get
\[\sum_{j=1}^{n-1}(-1)^j \csc^2 \left(\frac{j\pi}{2n}\right) = \frac{4n^2 - 4}{6} - n^2  - \frac{1}{2} ((-1)^n - 1) = -\frac{n^2}{3}  - \frac{1}{6} - \frac{1}{2}(-1)^n.\]
\end{proof}

\end{document}